\documentclass{article}

\usepackage{arxiv}

\usepackage[utf8]{inputenc} 
\usepackage[T1]{fontenc}    
\usepackage{hyperref}       
\usepackage{url}            
\usepackage{booktabs}       
\usepackage{amsfonts}       
\usepackage{nicefrac}       
\usepackage{microtype}      
\usepackage{lipsum}		
\usepackage{multirow}
\usepackage{graphicx}

\usepackage[square,sort,comma,numbers]{natbib}

\usepackage{doi}
\usepackage{wrapfig}
\usepackage[ruled,vlined]{algorithm2e}
\usepackage{enumitem}
\usepackage[dvipsnames]{xcolor}

\usepackage[utf8]{inputenc} 
\usepackage[T1]{fontenc}    
\usepackage{booktabs}       
\usepackage{microtype}      

\usepackage{latexsym}
\usepackage[framemethod=TikZ]{mdframed}
\usepackage{amsmath,mathtools,amsfonts,mathrsfs,amssymb,bm}
\usepackage{amsthm}
\usepackage{stmaryrd}
\usepackage{nicefrac}
\DeclareMathAlphabet{\mathbbm}{U}{bbm}{m}{n}
\usepackage{xspace}
\usepackage{enumitem}
\usepackage{picinpar}
\usepackage{xcolor,colortbl}
\usepackage{placeins}
\usepackage{hyperref}
\hypersetup{
    colorlinks=true,
    linkcolor=blue,
    citecolor=blue,
    urlcolor=blue,
}
\usepackage{subcaption,cleveref}

\newcommand{\vcol}[1]{\textcolor{blue}{#1}}

\def\rrr#1\\{\par
\medskip\hbox{\vbox{\parindent=2em\hsize=6.12in
\hangindent=4em\hangafter=1#1}}}

\newmdtheoremenv[innertopmargin=-2pt]{research}{Research direction}

\mdfdefinestyle{objstyle}{
    linecolor=black,            
    linewidth=1pt,       
    frametitlerule=true,
    backgroundcolor=blue!05,
}

          \newcommand{\cM}{\mathcal{M}}     \newcommand{\cX}{\mathcal{X}} \newcommand{\cR}{\mathcal{R}}  

\newcommand{\RR}{\mathbb{R}}

\newcommand{\err}{\text{Err}} 

\newcommand{\minimize}{\operatornamewithlimits{minimize}}

\newtheorem{theorem}{{\bf Theorem}}
\newtheorem{corollary}{{\bf Corollary}}

\definecolor{maincolor}{HTML}{032F99}
\definecolor{blue}{RGB}{31,64,122}
\definecolor{red}{HTML}{e05a87}

\addtocontents{toc}{\protect\setcounter{tocdepth}{0}} 
\title{Fairness Issues and Mitigations in (Differentially Private) Socio-Demographic Data Processes}

\date{} 					

\author{
 Joonhyuk Ko \\
  University of Virginia\\
  \texttt{tah3af@virginia.edu} \\
   \And
 Juba Ziani \\
  Georgia Institute of Technology\\
  \texttt{jziani3@gatech.edu} \\
  \And
 Saswat Das \\
  University of Virginia\\
  \texttt{duh6ae@virginia.edu} \\
  \AND
 Matt Williams \\
  RTI International\\
  \texttt{mrwilliams@rti.org} \\
  \And
 Ferdinando Fioretto \\
  University of Virginia\\
  \texttt{fioretto@virginia.edu} \\
}

\renewcommand{\shorttitle}{Fairness and Mitigation in (Differentially Private) Socio-Demographic Data Processes}

\hypersetup{
pdftitle={A template for the arxiv style},
pdfsubject={q-bio.NC, q-bio.QM},
pdfauthor={David S.~Hippocampus, Elias D.~Striatum},
pdfkeywords={First keyword, Second keyword, More},
}

\begin{document}
\maketitle

\begin{abstract}
Statistical agencies rely on sampling techniques to collect socio-demographic data crucial for policy-making and resource allocation. This paper shows that surveys of important societal relevance introduce sampling errors that unevenly impact group-level estimates, thereby compromising fairness in downstream decisions. To address these issues, this paper introduces an optimization approach modeled on real-world survey design processes, ensuring sampling costs are optimized while maintaining error margins within prescribed tolerances. Additionally, privacy-preserving methods used to determine sampling rates can further impact these fairness issues. This paper explores the impact of differential privacy on the statistics informing the sampling process, revealing a surprising effect: not only is the expected negative effect from the addition of noise for differential privacy negligible, but also this privacy noise can in fact reduce unfairness as it positively biases smaller counts. These findings are validated over an extensive analysis using datasets commonly applied in census statistics.
\end{abstract}


\section{Introduction}
\label{sec:introduction}

Statistical agencies across various countries gather, anonymize, and disseminate socio-demographic data, which is foundational to high-impact applications such as policy development, urban planning, and public health initiatives \cite{USCensusRedistricting, FHWA_CTPP_2023, Census_Health}. In the United States, for example, Census Bureau data guide more than \$2.8 trillion in federal funding yearly \cite{uscensus2023funding,TFHY:ijcai21}. Major surveys such as the American Community Survey (ACS) \cite{ACS_Bureau}, the Current Population Survey (CPS) \cite{CPS}, and the National Health Interview Survey (NHIS) \cite{NHIS} are central to gathering essential demographic data. The ACS, for example, annually collects data from approximately 3.54 million housing unit addresses across the United States (about a 1\% sample of the entire U.S. population). This sample-based approach allows the ACS to provide detailed insights into the population's living conditions, educational attainment, employment, and health status, among other factors. 
The accuracy of these demographic reports is thus crucial to ensure that resources and policy measures are effectively targeted toward appropriate population segments. 
However, despite their critical role, the collection of these statistics typically involves surveying a small fraction of the population, inherently introducing sampling errors.
While these surveys strive to provide estimates with controlled error rates and confidence intervals, such control is typically applied across the entire survey population. 
A key finding of this paper is to 
\emph{demonstrate how such an approach can lead to varying error rates among population groups}, particularly those distinguished by ethnicity, introducing biases in critical downstream tasks relying on this data. 

Therefore, the \emph{first} major contribution of this paper is to address the need for developing sampling schemes that not only aim to reduce costs but also meet acceptable errors \emph{within} each demographic group. The approach explored recasts the sampling process as an optimization program, ensuring that statistical accuracy is maintained across diverse sub-populations within prescribed confidence errors.

The \emph{second} major contribution of this work is to analyze the impacts of privacy-enhancing technologies on the biases of demographic data. In particular, we focus on \emph{differential privacy} (DP) \cite{dwork:06} as implemented by the U.S.~Census Bureau. Interestingly, contrary to prevailing intuition in the literature \cite{FHZ:ijcai22}, our findings suggest that these privacy measures do not necessarily exacerbate disparities. In fact, differential privacy can surprisingly mitigate the observed unfairness by boosting the representation of smaller populations (minorities) due to positive biases introduced during the DP post-processing phase. This observation provides a novel perspective on the tradeoffs between privacy and fairness, demonstrating that DP can contribute positively to fairness when implemented in contexts similar to those analyzed in this work. 
A schematic illustration of the proposed framework is provided in Figure \ref{fig:survey}.

\begin{figure*}
    \centering
    \includegraphics[width=0.95\linewidth]{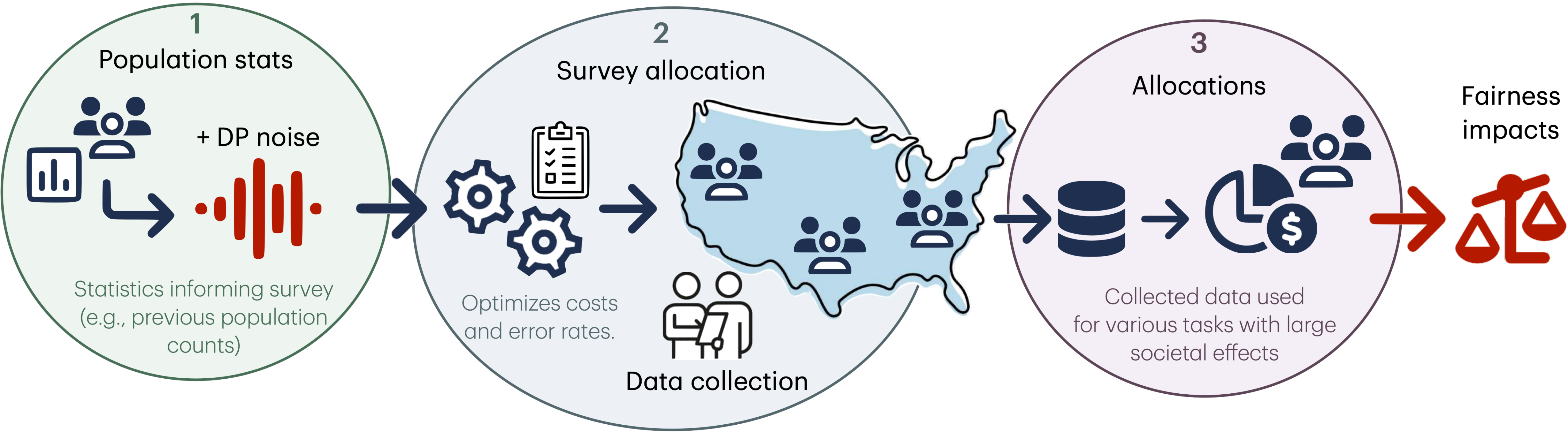}
    \caption{{\bf 1.} Population statistics from previous years are often used to inform the survey design process; Differential privacy can be used at this stage to protect sensitive information (e.g., population counts). 
    {\bf 2.} The survey process includes selecting the amount of the population to sub-sample as well as collecting information from individuals in multiple phases (e.g., phone calls and in-person interviews).
    {\bf 3.} The collected data is used for important tasks, such as the allocation of funds or the release of migration patterns. 
    {\em The paper studies the fairness impacts of this pipeline (steps 1 and 2) on multiple population segments}. 
    }
    \label{fig:survey}
\end{figure*}

\paragraph{Contributions.} 
This paper provides four key contributions:
\begin{enumerate}[leftmargin=*, topsep=0pt, parsep=0pt, itemsep=2pt]
\item First, we show that conventional sampling strategies may overlook potential disparate impacts on crucial ethnic demographic groups. This aspect, illustrated in more detail in Section \ref{sec:realworld_impact}, provides the basis for the proposed work. 

\item We then introduce an optimization approach aimed at mitigating these disparate errors, detailed in Section~\ref{sec:opt_sampling}. The proposed approach is modeled on real-world survey design processes, such as those employed by the ACS, which involve two phases: remote communications (e.g., phone calls, emails) and door-to-door, geographically targeted interventions. Our optimization framework is designed to optimize sampling costs while ensuring that error margins are within the prescribed error tolerance with a high probability for each population segment. 

\item Next, Section \ref{sec:sampling_privacy} explores the impact of differential privacy on the statistics informing the sampling process. Since the noise adopted by differential privacy can influence the estimation of group sizes, we ask if it may negatively affect the reliability of the error bounds established in our model and exacerbate unfairness. 
\emph{Surprisingly},  we found that on real U.S. survey data, not only is the impact of this noise negligible, but also due to an intriguing by-product of positive bias induced by DP post-processing on small counts, the resulting sampling process exhibits \emph{reduced} unfairness across various population segments. 

\item Finally, based on these observations and analyses, the paper introduces a heuristic, detailed in Section~\ref{sec:heuristic}, that approximates the fair decisions made by the proposed optimization problem. All the conducted experiments in the paper use real U.S.~census data.
\end{enumerate}

\section{Preliminaries and Goals}
\label{sec:preliminaries}

This paper considers a target population, such as the U.S.~population, segmented into $G$ distinct groups characterized by race, socio-economic status, and other demographic factors. Let $N$ represent the total population size, with $N_i$ indicating the size of each group $i \in [G]$. 
We examine the population statistics $\bm{\theta}(N)$, such as average income or poverty levels, and aim to estimate these via subsampling. The subsample, of size $n$ where $n \ll N$, is used to derive the estimates $\hat{\bm{\theta}}(n)$. In particular, this analysis extends to group-specific statistics, where $\hat{\bm{\theta}}(n_i)$ represents estimates from a subsample of size $n_i$ (the number of individuals sampled from group $i$), and $\bm{\theta}(N_i)$ represents the actual statistics for group $i$.

To ease notation, in the discussions that follow, $\bm\theta$ and $\hat{\bm\theta}$ will be used to represent the true population statistics and their estimates from the subsample, respectively, when clear from the context. Similarly, within each group $i \in [G]$, $\bm\theta_i$ and $\hat{\bm\theta}_i$ will denote the actual statistics for the population and their corresponding estimates from the subsample.

\paragraph{Accuracy and fairness.}
The accuracy of these estimates is evaluated through their error and variance, 
defined for group $i$ as 
\(
    \err(\hat{\bm \theta_i}) = |\hat{\bm \theta_i} - \bm\theta_i|,
\) 
and
\(
\mathrm{Var}(\hat{\bm\theta}_i) = \mathbb{E}[\hat{\bm\theta}_i^2] - (\mathbb{E}[\hat{\bm\theta}_i])^2,
\)
respectively. 
The primary goal is to devise sampling strategies that minimize the \emph{sampling cost}---defined in subsequent sections---while ensuring that the probability of an estimator's error exceeding a certain threshold ($\gamma_i$) for \mbox{each group $i$:}
\[
    \Pr(|\err(\hat{\bm \theta}_i)| > \gamma_i) \leq \alpha, \quad \forall i \in [G],
\]
remains less than $\alpha$.

\emph{Un}\textbf{fairness} in this context is quantified by the maximum discrepancy in estimator's variances between any two groups,
\[
    \xi_{\rm Var} = \max_{i,j \in G} | \mathrm{Var}(\hat{\bm\theta}_i) - \mathrm{Var}(\hat{\bm\theta}_j) |,
\]
since the goal of the survey process is controlling confidence intervals across various populations.

\smallskip
In the first part, this paper will discuss the development of optimal sampling schemes that balance the reduction of sampling costs against the constraints on estimator's accuracy for each group. Subsequently, we will examine the impact of privacy on the biases and variance of the various sub-populations, specifically when privacy-preserving counts $\tilde{N}_i$ are used instead of the actual group population counts $N_i$. We next discuss the notion of privacy adopted in this study.

\paragraph{Differential privacy.} Differential Privacy (DP) \cite{dwork:06} is a rigorous privacy notion that characterizes the amount of information of an individual's data being disclosed in a computation.
Formally, 
  a randomized mechanism $\cM:\cX \to \cR$ with domain $\cX$ and range $\cR$ satisfies $(\epsilon, \delta)$-\emph{differential privacy} if for any output $O \subseteq \cR$ and datasets $\bm{x}, \bm{x}' \in \cX$ differing by at most one entry (written $\bm{x} \sim \bm{x}'$),
  \begin{equation}
  \label{eq:dp}
    \Pr[\cM(\bm{x}) \in O] \leq \exp(\epsilon) \Pr[\cM(\bm{x}') \in O] + \delta.
  \end{equation}

\noindent
Intuitively, DP states that specific outputs to a query are returned
with a similar probability regardless of whether the data of any
individual is included in the dataset.  Parameter $\epsilon \!>\! 0$
describes the \emph{privacy loss} of the mechanism, with values close
to $0$ denoting strong privacy. When $\delta \!=\! 0$, mechanism $\cM$ is
said to achieve $\epsilon$- or pure-DP.

A function $f$ from a dataset $\bm{x} \!\in\! \cX$ to an output set 
$O \!\subseteq\! \RR^n$ can be made differentially private by injecting random noise onto its output. The amount of noise relies on the notion of \emph{global sensitivity} 
\(
\Delta_f \!=\! \max_{\bm{x} \sim \bm{x}'} \| f(\bm{x}) - f(\bm{x}') \|_p,
\) 
for $p \!\in\! \{1,2\}$. 
In particular, the \emph{Laplace mechanism} for histogram data release (sensitivity $\Delta_f\!=\!1$), defined by
\(
    \cM_{\text{Lap}}(\bm{x}) \!=\! \bm{x} + \text{Lap}(\nicefrac{1}{\epsilon}), 
\)
\noindent where $\text{Lap}(\eta)$ is the Laplace distribution centered at 0 and with scaling factor $\eta$, satisfies $(\epsilon, 0)$-DP.

\smallskip\noindent\textbf{Post-processing.} DP satisfies several important properties. Notably, 
\emph{post-processing immunity} ensures that privacy guarantees are
preserved by arbitrary post-processing steps. 
More formally, let $\cM$ be an $(\epsilon, \delta)$-DP mechanism and 
$g$ be an arbitrary mapping from the set of possible output sequences 
to an arbitrary set. Then, $g \circ \cM$ is $(\epsilon, \delta)$-differentially 
private.

\section{Real-World Impact: The ACS Case}
\label{sec:realworld_impact}

Next, the paper looks at the implications of sampling strategies in the American Community Survey (ACS), the largest sampling effort in the U.S. carried out by the U.S. Census Bureau. 
The ACS samples approximately 3.54 million housing unit addresses annually, representing about 1\% of the U.S. population, and of those approximately 1.98 million result in successful samples \cite{CensusSampleSize}. 
This relatively small sample size introduces inherent uncertainties, termed \emph{sampling errors}, which are critical in understanding the limitations and accuracies of the data collected. The Census Bureau addresses these uncertainties by calculating standard errors and publishing margins of error at a 90 percent confidence level.

The sampling process in data collection efforts such as the ACS introduces at least two fairness issues: \emph{disparate error rates} across different populations and \emph{disparate impact of privacy-preserving mechanisms} on sampling errors. Figure \ref{fig:motivational_figure} illustrates the former issue, and we will focus on the latter in Section \ref{sec:sampling_privacy}.
The figure shows the simulation results using 2021 data for Nebraska with a 1\% sampling rate obtained from IPUMS~\cite{IPUMS}. The lines represent the errors attained while estimating the \emph{population income}, within 6 distinct sub-populations (x-axis).
The red dotted horizontal lines illustrate the target error rates. 
Mimicking real-world behaviors, surveys are uniformly distributed so each group receives a proportionate number based on its size.

Notice that, while the overall population errors (rightmost bars, in dark-blue colors) are well within the prescribed confidence errors, 
minority groups, such as {\sl Native} (American Indians), {\sl Black}, and {\sl Asian}, experience systematically larger errors compared to the {\sl White} population. Additionally, when analyzing sub-populations by demographic groups, these discrepancies reveal that they often do not adhere to the 90 percent confidence levels established by the Census Bureau. 

\begin{figure}[t]
    \centering
    \includegraphics[width=0.6\columnwidth]{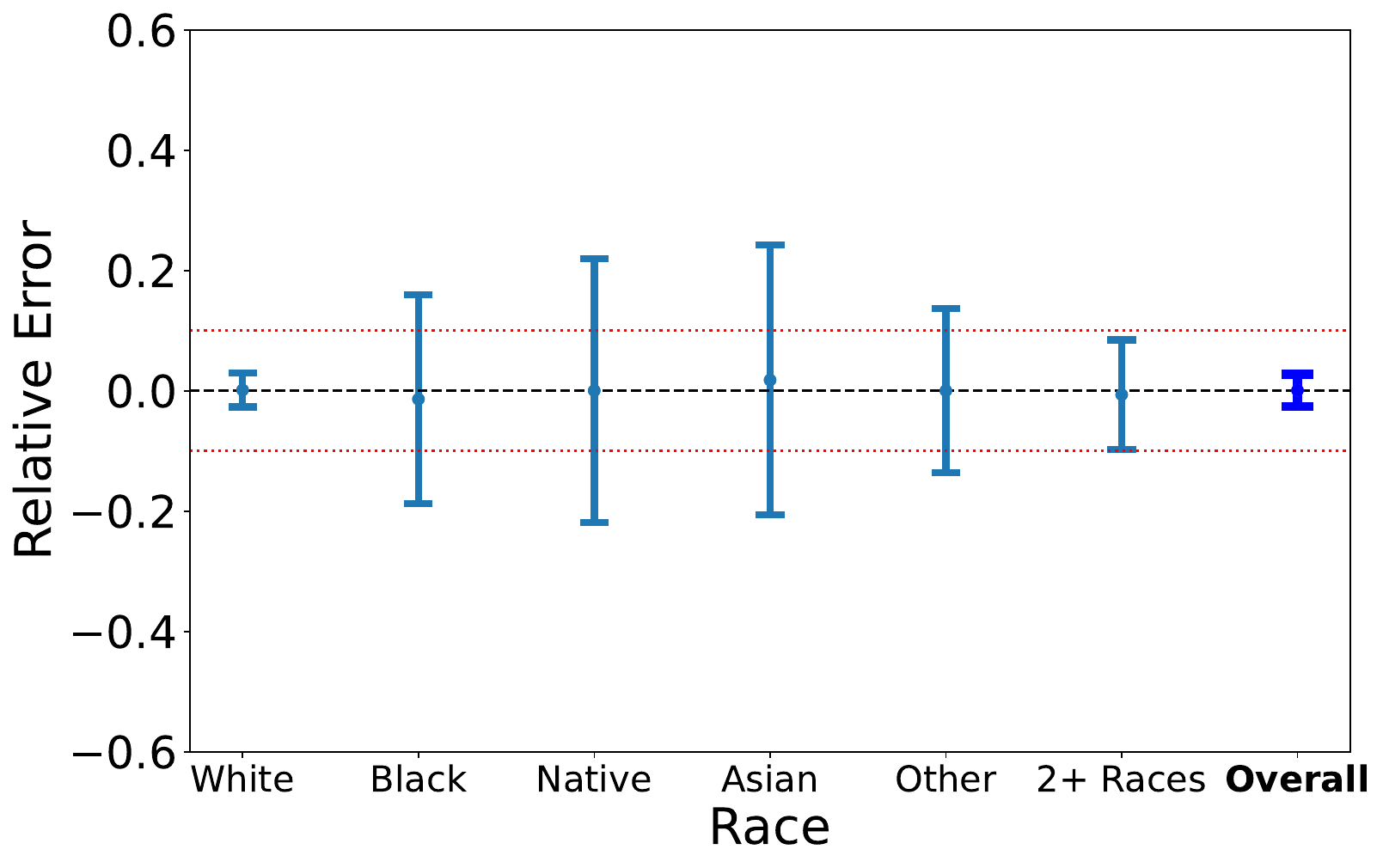}
    \caption{Disparate errors when allocating a proportional number of surveys to each racial group in Nebraska using 2022 ACS data. 
    2021 ACS data is used to compute the proportional allocation which subsample 1\% of the total population.  
    }
    \label{fig:motivational_figure}
\end{figure}

These disparities can have important repercussions given the role of these estimators in driving key policy decisions and beyond. 
Crucially, these behaviors are not well documented and the next section delves into our first key contribution: an optimization-based mitigation strategy.

\section{The Optimal Sampling Design Problem}
\label{sec:opt_sampling}

The proposed approach casts the sampling process for target estimation as an optimization process, going beyond typical cost minimization of large-scale surveys, to ensure that error rates within each sub-population are met with high probability. This section first outlines the optimization problem, considering real-world survey constraints, and then quantifies the error of estimators for each sub-population, enabling efficient implementation of the optimization model.

\subsection{Modeling Real-World Sampling Processes}

Large survey processes typically involve two phases. The first phase adopts various \emph{remote} data collection modes; for example, the ACS has used internet interviews since 2013 and computer-assisted phone interviewing until 2017 \cite{ACS_method}. 
The second phase relies on in-person, \emph{door-to-door interviews}, requiring the physical allocation of survey workers.  Although more expensive, this phase aims to improve data completeness and reliability, especially in environments where remote methods are less effective.

The efficacy of each phase is distinguished by distinct \emph{failure rates}-- the likelihood that an individual, once contacted, does not contribute data. These rates are denoted as $F^1_i$ and $F^2_i$ for the first and second phases respectively and vary across different population segments $i$. The costs associated with each contact attempt are denoted by $c_1$ and $c_2$ for the first and second phases, respectively. Typically, the cost-efficiency trade-off is clear: remote methods (Phase 1) are cheaper but often less effective ($F_i^1 > F_i^2$), while in-person interventions (Phase 2) yield higher success rates at a higher cost ($c_1 < c_2$). We define $g^r$ as the targeted or feasible sampling rate in region $r$ once selected for phase 2. Further, $\gamma_i$ represents the upper bound of acceptable error for population segment $i$ (i.e., error rates depicted by the red dotted lines in Figure \ref{fig:motivational_figure}), and $\alpha$ as the probability that this limit is exceeded. We define the following program to optimize this process:
{\normalsize
\begin{subequations}
    \label{eq:4}
    \begin{align}
        \minimize_{\vcol{\bm p, \bm z}} &\; 
         \underbrace{c_1 \left(\sum_{i \in [G]} \vcol{p_i} N_i \right)}_{1^{\text{st}} \text{ phase cost}}
              + \underbrace{c_2 \left(\sum_{r \in R} \vcol{z_r} \right)}_{2^{\text{nd}} \text{ phase cost}} \label{obj:4a}\\
    \hspace{-10pt}
            \texttt{s.t.}~~&
                {n}_i = \underbrace{\vcol{p_i} N_i \left(1-{F_i^1}\right)}_{1^{\text{st}} \text{ phase samples}}
                          +  \overbrace{\sum_{r\in R} \vcol{z_r} g^r N_i^r \left(1 - {F_i^2}\right)}^{2^{\text{nd}} \text{ phase samples}}
                         \; \forall i \in [G] \label{c:4b}\\
            &\; \Pr(|\err(\hat{\bm \theta}_i({n}_i))| > \gamma_i) \leq \alpha, \;\; \forall i \in [G], \label{c:4c}\\
            &\; 0 \leq \vcol{p_i} \leq 1 \;\; \forall i \in [G],
            ~\vcol{z_r} \in \{0, 1\} \;\; \forall r \in R. \label{c:4d}
\end{align}
\end{subequations}
}
The goal is to minimize the costs associated with contacting individuals through the two sampling phases described above, denoted by $c_1$ and $c_2$ (objective \eqref{obj:4a}), while ensuring that the error across each demographic group $i$ does not exceed $\gamma_i$ with a probability greater than $\alpha$ (constraint \eqref{c:4c}). Decision variables \vcol{$p_i$} model the fraction of group $i$ contacted in the remote Phase 1, and \vcol{$z_r$}, a binary variable, determines whether workers are deployed in region $r$ (where $R$ is the set of all regions) during Phase 2 (constraint \eqref{c:4d}). In the minimizer notation $\vcol{\bm p}$ and $\vcol{\bm z}$ are used as shorthand for the vectors $(\vcol{p_i})_{i \in [G]}$ and $(\vcol{z_r})_{r \in R}$, respectively.
Constraint~\eqref{c:4b} defines ${n}_i$, the average number of individuals that \emph{respond} to the survey across Phases 1 and 2. In population of size $N_i$, the surveyor contacts $\vcol{p_i} N_i$ individuals in Phase 1; $\vcol{p_i} N_i (1 - F_i^1)$ is then the rate at which individuals respond \emph{in expectation}. In Phase 2, in each region $r$, the surveyor reaches $\vcol{z_r}g^rN_i^r$ individuals, making $\vcol{z_r} N_i^r g^r (1 - F_i^2)$ the expected numbers of responses from population $i$ in region $r$.

\subsection{Tractable Error Quantification}

A key challenge with solving Program~\eqref{eq:4} is Constraint~\eqref{c:4c}, which involves a probability estimation. 
The lack of a closed-form expression for this probability hinders the direct integration of this constraint into the optimization. 
To address this, this section provides a tractable upper bound to be used in place of the probability in Constraint~\eqref{c:4c}.

Note that, using Chebyshev's inequality, the probability of the estimator's error exceeding $\gamma_i$ is bounded above by:
\begin{equation}
    \label{o1:bias_bound}
    \Pr(|\err(\hat{\bm \theta}_i)| > \gamma_i) = \Pr(|\hat{\bm \theta}_i - \bm \theta_i| > \gamma_i) \leq \frac{\sigma^2(\hat{\bm \theta}_i)}{\gamma_i^2}, 
\end{equation}
where $\sigma^2(\hat{\bm \theta}_i)$ represents the variance of the estimator $\hat{\bm \theta}_i$. This variance can then be estimated empirically, as done in practice \cite{ACS_method}, using prior data releases. This creates a statistical proxy, which is discussed in the Section~\ref{sec:empirical_variance_estimation}. 

For a given confidence level $\alpha$, from \eqref{o1:bias_bound}, we can replace Constraint~\eqref{c:4c} by the \emph{stronger} constraint $\frac{\sigma^2(\hat{\bm \theta}_i)}{\gamma_i^2} \!\leq\! \alpha$, and obtain a closed-form approximation for the threshold $\gamma_i$ as: 
\begin{align}\label{c:4c_stronger}
    \sigma^2(\hat{\bm \theta}_i) \leq \alpha \gamma_i^2.
\end{align}
This new constraint strengthens the program by enforcing $\Pr(|\err(\hat{\bm \theta}_i)| > \gamma_i) \leq \frac{\sigma^2(\hat{\bm \theta}_i)}{\gamma_i^2} \leq \alpha$, which restricts the likelihood that the error in group $i$ exceeds the desired $\gamma_i$ threshold, thus tightening the optimization.
The variance of the estimator $\sigma^2(\hat{\bm \theta}_i)$ can thus be expressed as 
\begin{equation}
    \label{eq:emp_variance}
    \sigma^2(\hat{\bm \theta}_i) = \frac{C_i}{n_i},
\end{equation}
where $C_i$ is a constant that depends on the variance for group $i$. 
This follows from the variance of the estimator $\sigma^2(\hat{\bm \theta}_i)$ being inversely proportional to the sample size $n_i$.\footnote{For $n$ iid random variables $x_i\!\sim\!\mathcal{N}(\mu,\sigma^2)$ and estimator $\hat x \!=\! \frac{\sum_i^n x_i}{n}$, $\mathrm{Var}(\hat x) \!=\! \mathrm{Var}(\frac1n\sum_i^n \!x_i) \!=\! \frac1{n^2}\mathrm{Var}(\sum_i^n \!x_i) \!=\! \frac1{n^2} n\sigma^2 \!=\! \frac{\sigma^2}{n}\!\!$.}
Thus, by substituting this expression into Equation \eqref{c:4c_stronger}, constraint \eqref{c:4c} can be replaced by the following tractable form: 
\begin{equation}
\label{c:4c_approx}
    n_i \geq \frac{C_i}{\alpha \gamma_i^2}, \quad \forall i \in [G] \tag{$2\bar{c}$}.
\end{equation}

\begin{figure*}[t]
    \centering
    \begin{subfigure}[t]{0.3\textwidth}
        \centering
        \includegraphics[width=\linewidth]{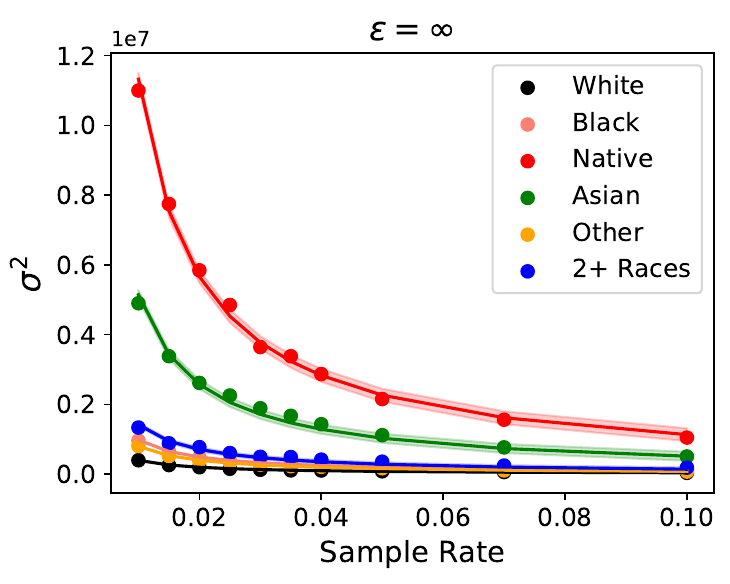}
    \end{subfigure}
    \hfill
    \begin{subfigure}[t]{0.3\textwidth}
        \centering
        \includegraphics[width=\linewidth]{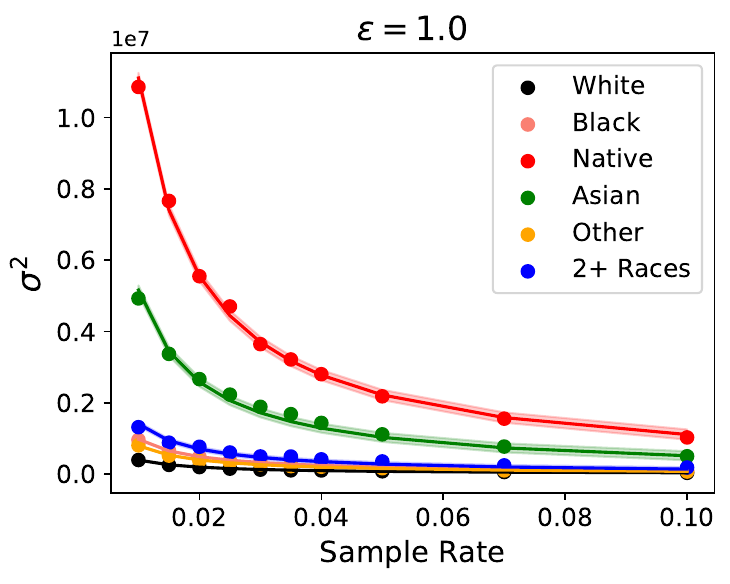}
    \end{subfigure}
    \hfill
    \begin{subfigure}[t]{0.3\textwidth}
        \centering
        \includegraphics[width=\linewidth]{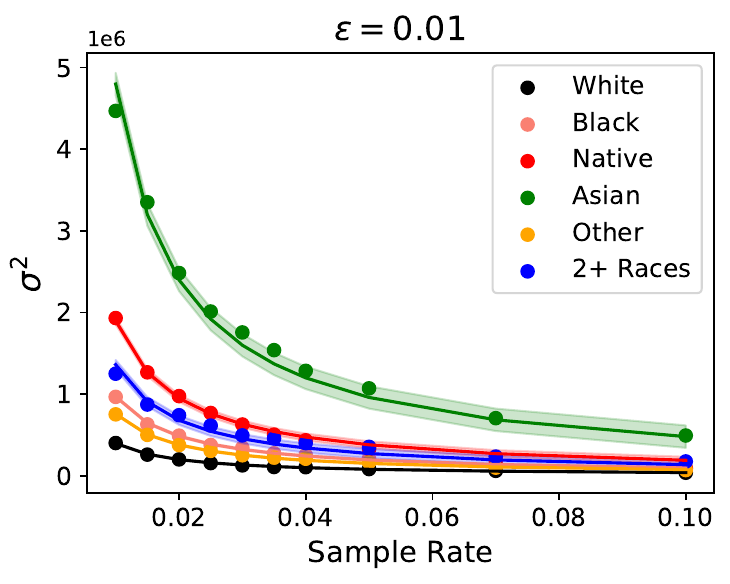}        
    \end{subfigure}
    \caption{Estimating the variance of mean \emph{income} in Connecticut using \emph{race} as a subgroup with different privacy budget $\varepsilon$. 
    {\bf Points}: actual estimator measurement, {\bf curves}: proxy function fitting. Results averaged over 200 trials and 200 data points.}
    \label{fig:proxy}
\end{figure*}

\subsection{Empirical Variance Estimation}
\label{sec:empirical_variance_estimation}

In the above expression, $C_i$ is a constant that depends on the variance of the population for group $i$ (see Equation \eqref{eq:emp_variance}) and can be estimated by approximating the variance of the target estimates across a range of sampling rates. Figure \ref{fig:proxy} (left) illustrates such an approach 
showing how the variance of each subgroup changes across sampling rates within practical (e.g., budget) constraints that limit sampling to no more than 10\% of the population. 
The middle and right figures also report this effect when privacy is considered to protect the sub-population counts $N_i$, as discussed in Section \ref{sec:sampling_privacy}. 

Approximating the variance $\sigma^2(\bm{\theta}_i)$ relies on fitting a curve of the form $\frac{a_i}{x}$ for each population group $i$, where $a_i$ is the constant to be estimated and $x = \frac{n_i}{N_i} \in [0,1]$ represents the sampling rate. These curves, referred to as \emph{proxy functions}, estimate the variance based on sampling rates rather than absolute sample sizes, enabling a direct integration in our error quantification constraints.
Finally, to translate these proxy functions into a usable format within the constraint \eqref{c:4c_approx}, we equate $C_i$ with $a_i\,N_i$, following:

\[
    \sigma^2(\hat{\bm \theta}_i) = \frac{C_i}{n_i} = \frac{a_i}{x} = \frac{a_i}{n_i/N_i} = \frac{a_iN_i}{n_i}.
\]

\section{Private Sampling Scheme}
\label{sec:sampling_privacy}

The sampling design discussed above assumes accurate knowledge of 
population sizes; however, the confidentiality of collected micro-data is often legally mandated. For example, it is regulated by Title 13~\cite{USTitle13} in the U.S., and, to comply with it, the U.S.~Census Bureau used differential privacy for their 2020 decennial census release~\cite{USCensusAdoptsDP}.
However, the use of DP mechanisms introduces perturbations in the data that may disproportionately affect smaller populations \cite{TFHY:ijcai21,ZFH:ijcai22}. 
This section studies how privacy-protected statistics could influence fairness in data collection. 

More precisely, we consider population sizes $N_i^r$ released differentially-privately for each group $i \in [G]$ and region $r \in R$, emulating the census data release. Therefore, instead of having access to the exact $N_i^r$, the survey designer only has access to imperfect, noisy estimates given by:
\begin{align}\label{eq:noisy_N}
\tilde{N}_i^r = \max\left(0, N_i^r + Lap(\Delta x/\varepsilon)\right),
\end{align}
where $\Delta x \!=\! 1$ (sensitivity of the count query).
We further note that the noisy counts are post-processed to ensure non-negativity (as is done in the U.S.~Census \cite{Spence_2023}), here using the $\max(0,.)$ operator.
    
The challenge in this context is that noise can distort estimates of population sizes $N_i^r$, influencing the number of individuals $n_i$ who respond to the survey in each group $i$. This distortion affects errors in Constraint~\eqref{c:4b} and compromises achieving the desired error targets and confidence levels. This section outlines our second key contribution: we offer theoretical insights into the biases introduced by using $\tilde{N}_i^r$ instead of $N_i^r$, while Section~\ref{sec:results} will offer a practical analysis of these impacts. 
Our main result is a \emph{closed-form} expression for the bias of the estimate $\tilde{N}_i^r$, showing that this bias is invariably positive.
\begin{theorem}
\label{thm:bias}
For all $i \in [G]$, $r \in R$, the bias of estimate $\tilde{N}_i^r$ is given in closed-form by:
\[
\mathcal{B}(\tilde{N}_i^r) = 
\mathbb{E} \left[\tilde{N}_i^r \right] - N_i^r = 
\frac{\Delta x}{2 \varepsilon} \exp \left(-\frac{N_i^r \varepsilon}{\Delta x}\right) > 0.
\]
\end{theorem}

\begin{proof}[Proof of Theorem \ref{thm:bias}]
\label{proof:bias}
Let $f(z) = \frac{1}{2b} \exp \left( - |z-N_i^r|/b \right)$ be the 
pdf of the Laplace centered around $N_i^r$ with 
$b = \Delta x/\varepsilon$. The expectation of the post-processed count is then given by:

\begin{align}
\mathbb{E} \left[\tilde{N}_i^r \right] 
& = \int_{-\infty}^\infty \max\left(0,z\right) f(z) dz
\\&= \int_{-\infty}^{0} 0~dz + \int_{0}^{N_i^r} z f(z)dz + \int_{N_i^r}^{\infty} z f(z)dz. \label{thm1:8}
\end{align}

\noindent The following computes separately the three terms in Equation \eqref{thm1:8}:

\begin{align}
\int_{-\infty}^{0} 0~dz &= 0 \label{thm1:9}
\\ \int_{0}^{N_i^r} z f(z)dz &= \frac{1}{2}(N_i^r - b) - \frac{1}{2}(-b)\exp(\frac{-N_i^r}{b}) \label{thm1:10}
\\ \int_{N_i^r}^{\infty} z f(z)dz &= \frac{1}{2}(N_i^r + b). \label{thm1:11}
\end{align}
Combining equations \eqref{thm1:9} - \eqref{thm1:11} with \eqref{thm1:8} gives:
\[ \mathbb{E} \left[\tilde{N}_i^r \right] = N_i^r + \frac{b}{2} e^{-N_i^r/b}. \]

\noindent Then, the bias on $\tilde{N}_i^r$ is
\[
\mathcal{B}(\tilde{N}_i^r) = \mathbb{E} \left[\tilde{N}_i^r \right] - N_i^r = \frac{b}{2} e^{-N_i^r/b} > 0.
\]
\end{proof}

Observe that not only is the bias term always positive ($\varepsilon$ is always positive), but also for fixed privacy budget $\varepsilon$, this bias is \emph{higher} on groups with small $N_i^r$. This implies that minority populations, such as Native Americans, are more likely to be overestimated. Further, if all other parameters are fixed, when $\varepsilon$ increases, the bias \emph{decreases}; in extreme cases, when $\varepsilon \to \infty$ (no privacy), the bias converges to $0$, and when $\varepsilon \to 0$ (perfect privacy), the bias grows large. This implies an interesting effect: the bias-induced overestimation of minority populations and its beneficial effects in correcting for under-allocations increase as $\varepsilon$ decreases.
Surprisingly, and contrary to much of previous known effects \cite{FHZ:ijcai22}, in the context studied here, enforcing stronger privacy induces \emph{less unfairness towards minority populations!}

This result implies the following corollary deriving the bias of the aggregated (e.g., state level) counts on $\tilde{N}_i \!=\! \sum_{r\in R} \tilde{N}_i^r$ given the various $\tilde{N}_i^r$ (e.g., county level), highlighting once again a more pronounced effect on minority populations:
\begin{corollary}
\label{cor:bias_aggr}
    The bias of the aggregated counts for each subgroup on the \emph{state} level is
    \[
    \mathcal{B}(\tilde{N}_i) = \mathbb{E} \left[\tilde{N}_i \right] - N_i = 
    \sum_{r \in [R]} \frac{\Delta x}{2 \varepsilon} \exp\left(\frac{-N_i^r \varepsilon}{\Delta x} \right) > 0.
    \]
\end{corollary}

\begin{proof}[Proof of Corollary \ref{cor:bias_aggr}]
\label{proof:bias_aggr}
The expectation of the post-processed count on the \emph{state level} is given by:
\begin{align}
\mathbb{E} \left[\tilde{N}_i \right] 
&= \mathbb{E} \left[\sum_{r \in [R]} \tilde{N}_i^r \right] 
\\ &= \sum_{r \in [R]} \mathbb{E} \left[ \tilde{N}_i^r \right] (\text{Linearity of Expectation})
\\ &= \sum_{r \in [R]} \left( N_i^r + \frac{b}{2} e^{-N_i^r/b} \right) (\text{Theorem~\ref{thm:bias}})
\\ &= N_i + \sum_{r \in [R]} \frac{b}{2} e^{-N_i^r/b}.
\end{align}

\noindent Then, the bias of the post-processed count on the \emph{state level} can be expressed as 
\[
\mathcal{B}(\tilde{N}_i) = \mathbb{E} \left[\tilde{N}_i \right] - N_i = \sum_{r \in [R]} \frac{b}{2} e^{-N_i^r/b} > 0.
\]
\end{proof}
This further highlights that while this positive bias will have a major relative impact on minority populations, overestimating minority populations allows a standard sampling scheme to allocate more surveys to minorities, thus reducing their relative errors.
An empirical analysis of this phenomenon is provided in Section \ref{sssec:DP-sampling}.

\section{Experimental Results}
\label{sec:results}

Next, the paper provides empirical evidence for the efficacy of the proposed optimization method on real-world data and settings first without and then with privacy considerations at hand. The experiments examine survey costs, group fairness, and utility offered by the proposed fairness-aware method. 

\smallskip\noindent\textbf{Datasets and settings.}
The experiments use ACS data from IPUMS \cite{IPUMS} for 2021 and 2022, leveraging 2021 data for estimating the various $N_i^r$ and 2022 data as ground truth for sampling and assessing target estimators. We divide geographical units based on Census Tract-level data, each containing about 4,000 individuals (data processing and setup details are in Appendix \ref{app:experimental-details}). The focus is on estimating annual total pre-tax personal income across different ethnic and educational groups as defined by IPUMS and the Census Bureau. 
We focus on Connecticut as the primary state for analysis here and relegate additional 
comprehensive results for other states and demographic breakdowns by education levels to Appendix \ref{app:other_states} and  \ref{app:education_level}, respectively.

\smallskip\noindent\textbf{Algorithms.} 
This analysis evaluates various survey allocation mechanisms, comparing their  efficiency, fairness, and effectiveness in achieving desired confidence levels, not only at the entire state level but also at the sub-population levels:
\begin{itemize}[leftmargin=*, parsep=0pt, itemsep=0pt, topsep=2pt]
    \item \textbf{Standard Allocation:}
    This baseline method, also known as \emph{proportional stratified random sampling}, allocates surveys to each population group $i$ in proportion to their size. This approach is a \emph{stronger} baseline than simple random sampling for two key reasons: it provides more precise population estimates by reducing variance within each subgroup\footnote{This is because there is no variance on the number of surveys each subgroup receives.}, and it ensures that the sample proportions are representative of the overall population \cite{Lohr_1998}.

    \item \textbf{Optimization: Phase 1 Only:} This variant applies 
    the optimization from Program \eqref{eq:4}, assuming the survey is 
    conducted by only using the first phase. 
    More concretely, the program excludes the \emph{2$\text{nd}$ phase} components in \eqref{obj:4a} and \eqref{c:4b} (for additional details see Appendix \ref{sec:algorithm}).
    This model mirrors the proportional stratified random sampling by optimizing survey numbers within a single operational phase.

    \item \textbf{Optimization: Phase 1 and 2:} This approach uses both 
    phases as outlined in optimization \eqref{eq:4}, aligning closely 
    with practical survey methodologies. 
    Note that the choice of failure rates and costs influences the optimization outcomes. In particular, high failure rates ($F_i^1$ and $F_i^2$) or low error tolerances ($\alpha$ and $\gamma$) increase the total survey cost due to more failures and tighter constraints. Details on the effects of varying these parameters are discussed in a comprehensive ablation study in Appendix \ref{sec:varying_param}.
    The default confidence constraints are set at $\alpha = 0.1$ and $\gamma_i =$ 10\% of the mean income for each subgroup $i$. Default failure rates are $F_i^1 = 0.60$ and $F_i^2 = 0.20$ $\forall i\in[G]$, and the cost of surveying a region in phase 2 is set to be 500 times more expensive than the cost of reaching 
    out to an individual with phone calls in phase 1. Finally, the sampling rates for geographies are set as $g^r = 0.1$ $\forall\, r \in R$. This translates to sampling 400 people per selected region.
\end{itemize}

\begin{figure}[t]
    \centering
    \includegraphics[width=0.6\columnwidth]{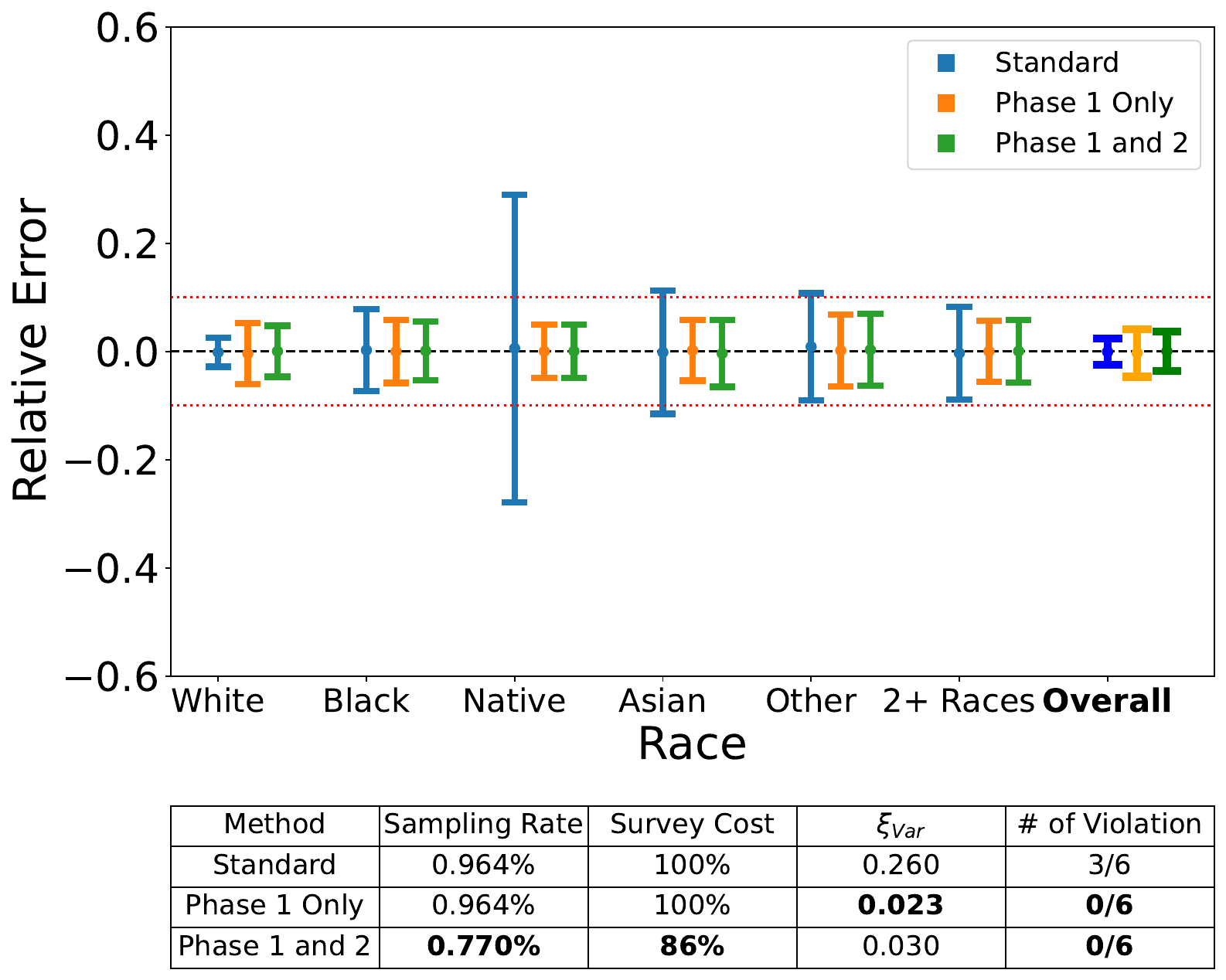}
    \caption{Relative group errors from estimating mean income in Connecticut.} 
    \label{fig:relative_error_no_privacy}
\end{figure}

\smallskip\noindent\textbf{Evaluation metrics.}
The evaluation of these mechanisms focuses on three primary metrics:
\begin{enumerate}[leftmargin=*, parsep=0pt, itemsep=0pt, topsep=2pt]
\item \emph{Survey cost}: Measured as a percentage of the cost reference 
      used by the standard allocation.
\item \emph{Fairness of variance}: Assesses the equitable distribution of 
    survey errors across different groups.
\item \emph{Confidence compliance}: Evaluates the ability to meet the 
    prescribed confidence errors ($\gamma_i$) at a 10\% threshold, 
    which aligns with the current standards of the ACS \cite{ACS_method}, and setting $\alpha = 0.1$.
\end{enumerate}
Further exploration of the impact of privacy on these 
metrics is detailed in Section \ref{sssec:DP-sampling}.\textbf{}

\subsection{Optimized Sampling: Errors and Fairness\!\!\!}

We start by assessing the performance of two variants of our method (``\emph{Phase 1 Only}'' vs.~``\emph{Phase 1 and 2}'') against the standard allocation mechanism,  without DP considerations. 

The results, summarized in Figure~\ref{fig:relative_error_no_privacy}, show that the \emph{Standard Allocation} method yields the lowest variance of error when estimating the overall population's income. However, this method disproportionately affects minorities, who receive fewer surveys and experience a higher variance of error at the group level. This discrepancy results in the worst fairness of variance ($\xi_\text{Var}$) observed (refer to table under Figure~\ref{fig:relative_error_no_privacy}), and minority groups even fail to meet the confidence constraints set for their estimations!

In contrast, the \emph{Phase 1 Only} optimization approach achieves a more uniform error variance across all subgroups while using \emph{the same} budget used in the \emph{Standard Allocation} method. Inspecting the optimization solutions, it can be observed that equity is achieved by allocating a similar number of surveys to each subgroup, irrespective of their population size. Figure~\ref{fig:allocation_plot} reports the number of survey allocations by race and by each method, and provides a clear view of the nature of the disparities. 
This redistribution significantly lowers the variance of the errors for minorities (including Native, Black, Asian, Other, and 2+), while slightly increasing it for the majority (White). Importantly, this approach \emph{enhances fairness and ensures all groups meet the confidence thresholds}, addressing the main drawback of the \emph{Standard Allocation}. In Section~\ref{sec:heuristic}, we provide further discussion on why distributing surveys equally across subgroups is an effective alternative. 

\begin{figure}[t]
    \centering
    \includegraphics[width=0.55\columnwidth]{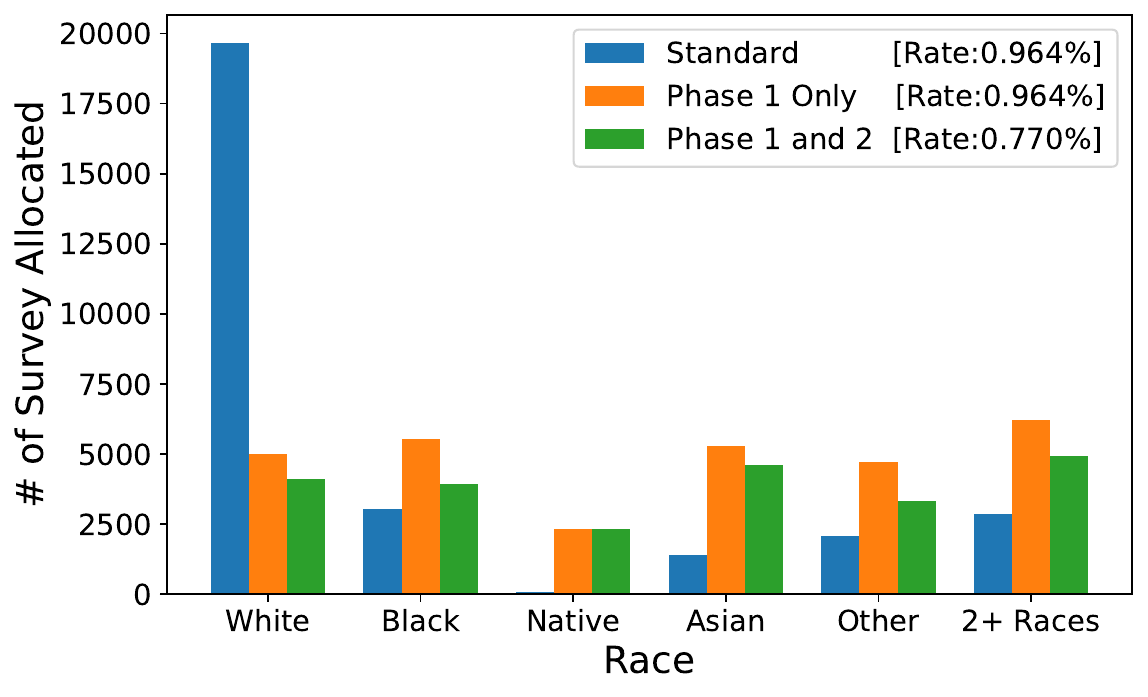}
    \caption{Number of surveys allocated for each subgroup in the experiments reported in Figure~\ref{fig:relative_error_no_privacy}.} 
    \label{fig:allocation_plot}
\end{figure}

\begin{figure*}[!t]
    \centering
    \includegraphics[width=1.0\textwidth]{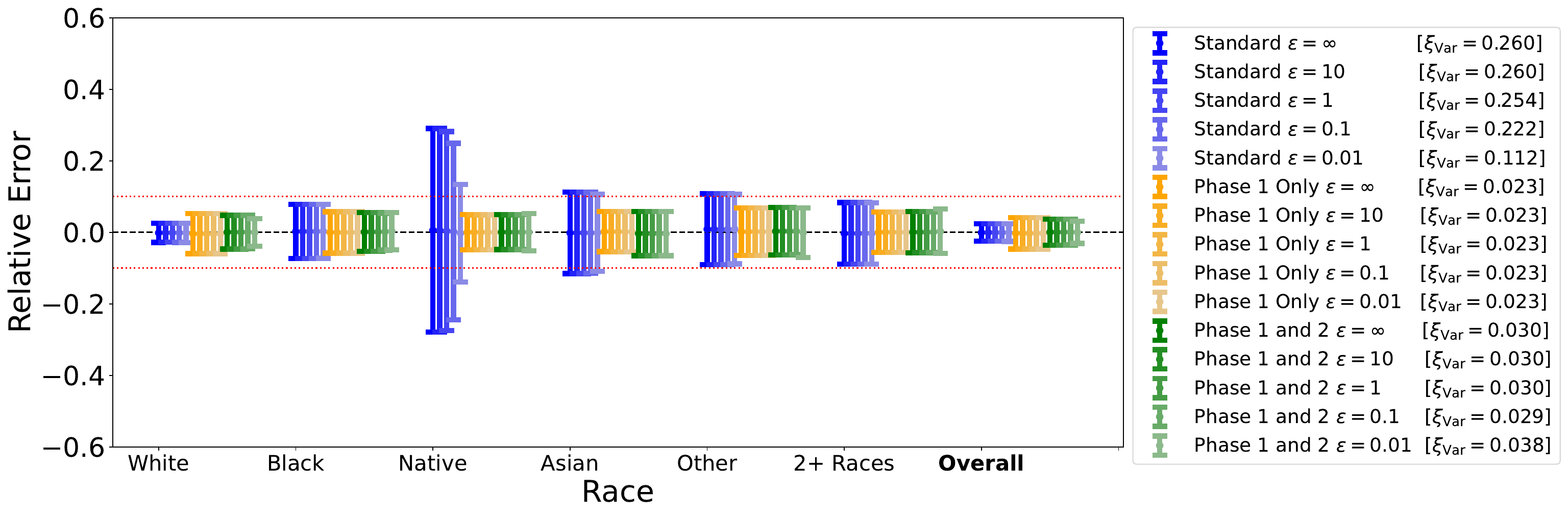}
    \caption{Relative errors from estimating mean income using DP-noised $N_i^r$ in Connecticut. Each region used in the phase 2 contains approximately 4,000 people, similar to the size of Census Tracts.
    }
    \label{fig:relative_error_with_privacy}
\end{figure*}

Next, we focus on our main approach. As discussed in Section~\ref{sec:opt_sampling}, phase 2 is characterized by a higher success rate ($F_i^1>F_i^2$) at a greater cost ($c_1<c_2$). Despite its higher per-survey cost, phase 2's low failure rate results in a higher number of \emph{successful} samples for the same overall cost, making the \emph{Phase 1 and 2} method \emph{substantially cheaper} (86\% of Phase 1 Only cost) (see table under Figure~\ref{fig:relative_error_no_privacy}). However, once regions are selected for phase 2, simple random sampling is executed at a 10\% rate ($g^r$) from each chosen region. This method introduces some uncertainty in the number of successful samples for each subgroup, although the optimizer prioritizes regions with high densities of the targeted population. This slightly reduces the performance and fairness of variance compared to the \emph{Phase 1 Only} optimization. Nonetheless, the \emph{Phase 1 and 2} method meets (by construction) the confidence constraints for every group, as empirically demonstrated.

\begin{table}[t]
\centering
\resizebox{0.60\columnwidth}{!}{
    \begin{tabular}{|c|c|c|c|c|c|c|c|}
    \hline
    $\varepsilon$ \symbol{92} Race & White      & Black   & \textbf{Native}  & Asian   & Other   & 2+ Races  & Total        \\ \hline
    $\infty$                       & 2,039,731  & 315,568 & \textbf{7,571}   & 143,584 & 215,150 & 295,844   & 3,017,448    \\ \hline
    $10$                           & 2,039,355  & 315,182 & \textbf{7,524}   & 143,226 & 214,766 & 295,482   & 3,015,535    \\ \hline
    $1$                            & 2,039,298  & 315,199 & \textbf{7,699}   & 143,260 & 214,736 & 295,495   & 3,015,687    \\ \hline
    $0.1$                          & 2,038,844  & 315,320 & \textbf{10,681}  & 143,846 & 214,555 & 295,705   & 3,018,951    \\ \hline
    $0.01$                         & 2,034,218  & 321,513 & \textbf{42,068}  & 158,498 & 222,160 & 304,533   & 3,082,990    \\ \hline
    \end{tabular}}
\caption{Impact of DP on estimated population size for each race in Connecticut using prior (e.g., ACS 2021 dataset).}
\label{table:population_estimation}
\end{table}

\subsection{DP-Sampling: Errors and Fairness}
\label{sssec:DP-sampling}

Next, we focus on the setting \emph{with differential privacy}, employing the privately adjusted counts $\tilde{N}_i^r$ as described in Equation~\eqref{eq:noisy_N}. 
The main results are reported in Figure~\ref{fig:relative_error_with_privacy}, again for the state of Connecticut, and additional results analyzing other states are reported in Appendix \ref{app:other_states}.

The first surprising result comes when analyzing the \emph{Standard Allocation} approach. While one might expect the added noise to exacerbate errors for minorities, here {\em adding more noise reduces the variance of errors for minorities}! This interesting behavior occurs because the induced strong positive bias overestimates the minority population size resulting in a higher allocation of surveys to these groups, 
as observed by our theoretical analysis in Section \ref{sec:sampling_privacy}.
Table~\ref{table:population_estimation} summarizes this effect, where it is possible to observe how much the smallest group (\emph{Native}) size is conflated with the addition of noise (smaller $\varepsilon$).
This increased allocation not only reduces the error variance but also improves fairness, countering the typical expectation that more noise increases error.

On the other hand, the \emph{Phase 1 Only} optimization appears to be insensitive in the variance of errors with respect to $\varepsilon$. This stability arises due to $C_i$, which determines the required number of samples, does not depend on group size. Thus, noise added to the population count does not impact survey distribution, maintaining consistent error variance across varying noise levels.

In contrast, the \emph{Phase 1 and 2} method experiences slight changes in the variance of errors with added noise. This effect is due to how noise affects the selection of regions for Phase 2, which relies on the population composition from prior data. More noise increases the probability of incorrect region selection, altering survey distribution and consequently, error variance. 

The observed higher positive bias in minorities as $\varepsilon$ decreases is explained by Corollary~\ref{cor:bias_aggr}, which notes that a smaller $N_i^r$ leads to larger positive biases. This implies that the region size used in phase 2 directly influences the level of positive bias. A comprehensive analysis of how different region sizes impact this bias is provided in Appendix \ref{app:sparsity_analysis}. 

Finally, note that the larger positive bias observed as $\varepsilon \rightarrow 0$ in Table~\ref{table:population_estimation} led to less discrepancy between the population sizes of different subgroups by overestimating the minorities. This results in a more uniform allocation of surveys across subgroups, thereby improving fairness. This suggests that allocating an equal number of surveys to each subgroup, irrespective of population size, may result in roughly the same relative errors. 
We implement this simple \emph{heuristic} and report a detailed analysis in Section~\ref{sec:heuristic}.
Although suboptimal (and more costly than our proposed optimization approach), this heuristic achieves a much-improved fairness score compared to the \emph{Standard Allocation} approach.

\section{Heuristics-Guided Sampling Scheme}
\label{sec:heuristics}
\label{sec:heuristic}
We will now introduce a \emph{heuristic} approach inspired by the findings from Section~\ref{sssec:DP-sampling}, where we observed that allocating a closer to equal number of surveys to each subgroup resulted in similar relative errors. The relative standard error (RSE) can be  calculated as:
\[  \text{RSE} = \left( \frac{\sigma}{\sqrt{n}} \right)/\mu = \frac{\sigma / \mu}{\sqrt{n}}. \]

Here, $\sigma$ represents the income standard deviation, $\mu$ is the average income, and $n$ denotes the sample size. The term $\sigma / \mu$, known as the coefficient of variation, is generally consistent across different subgroups. To achieve equal RSE, we need to ensure that $n$ is similar across subgroups. If $\sigma / \mu$ is constant across groups, equal sample sizes would result in a fairness score $\xi_\text{Var} = 0$, indicating perfect fairness. However, if $\sigma / \mu$ varies significantly among groups, equal sample sizes will still lead to a lower fairness score. We calculated the coefficient of variation for each race in Nebraska, as shown in Table~\ref{tab:cov_NE}.

\paragraph{Improved Motivation Figure.} 
Referring to our motivational figure (Figure~\ref{fig:motivational_figure}), we observe a significant improvement in the fairness score by allocating the same number of surveys to each race (i.e., \emph{heuristic} approach), as illustrated in Figure~\ref{fig:motivational_figure_heuristic}. This suggests that surveyors can apply the heuristic method with the same sampling rate used in the \emph{Standard Allocation} method to enhance fairness without relying on the optimization algorithm. Notably, the variance of error for Native Americans is smaller compared to other groups, despite the equal allocation, which can be attributed to their smaller coefficient of variation in Nebraska, as detailed in Table~\ref{tab:cov_NE}.


\begin{figure}[h]
    \centering
    \includegraphics[width=0.55\columnwidth]{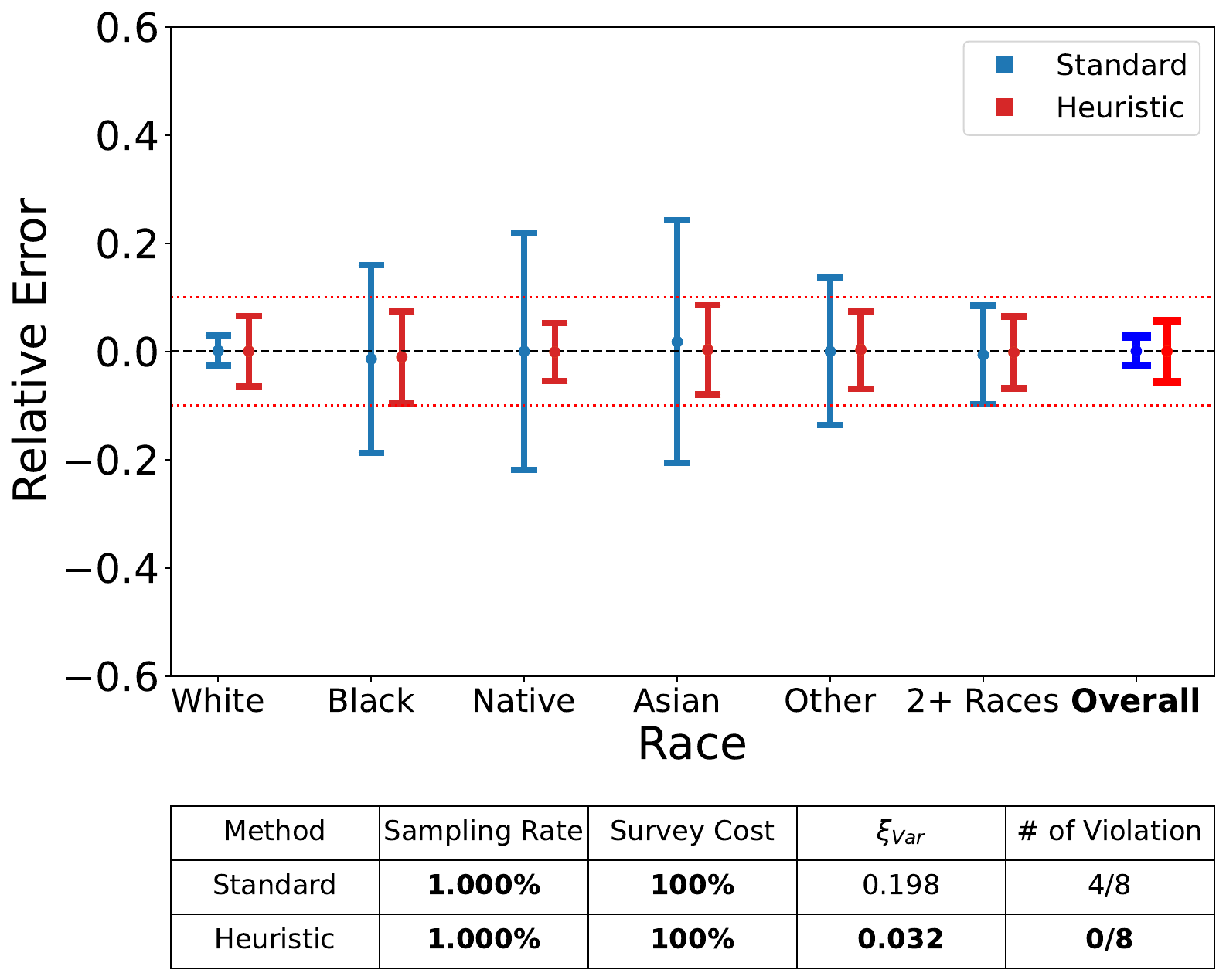}
    \caption{Relative errors from estimating income in Nebraska using the \emph{heuristic} method at 1\% sampling rate.}
    \label{fig:motivational_figure_heuristic}
\end{figure}

\begin{table}[h]
\centering
\resizebox{0.55\columnwidth}{!}{
    \begin{tabular}{|c|c|c|c|c|c|c|c|}
    \hline
    White  & Black & Native & Asian & Other & 2+ Races & Overall  \\ \hline
    1.33   & 1.72  & 1.13   & 1.57  & 1.30  & 1.19     & 1.36    \\ \hline
    \end{tabular}}
\caption{Coefficient of variation ($\frac{\sigma}{\mu}$) of income for each race in Nebraska using the ACS 2022 dataset.}
\label{tab:cov_NE}
\end{table}


\paragraph{Runtime Comparison.} 

Note that the \emph{heuristic} approach has a constant time runtime complexity since given a sampling rate $\rho$, determining $n_i \triangleq \frac{\rho \sum_{i \in [G]}N_i}{|G|}$ is a constant-time operation. In contrast, the optimization approach is a mixed-integer linear program (MILP) --- it involves binary decision variables for geographic selection in phase 2. While MILPs are NP-hard in general, the practical runtime for this algorithm in all our experiments was less than 3 seconds (for approximately 1500 regions) due to the high sparsity of the constraint matrix (low dependency between region selections).

\section{Conclusion}
This work was motivated by our observations of unfairness in large survey efforts of critical importance for driving many policy decisions and allocations of large amounts of funds and benefits. 
This paper showed that in surveys like the American Community Surveys, traditional sampling methods disproportionally affect minority groups, leading to biased statistical outcomes. 
To address these issues, we introduced an optimization-based framework to ensure fair representation in error margins in each population segment while minimizing the total sampling costs. Additionally, this paper examined the effects of differential privacy on the accuracy and fairness of the realized surveys. Contrary to common intuitions, our findings reveal that differential privacy can reduce unfairness by introducing positive biases beneficial to underrepresented populations. 
These findings are validated through rigorous and comprehensive experimental analysis using real-world data, demonstrating the effectiveness of the proposed optimization-based strategies in terms of enhancing fairness without compromising data utility and costs. 

We believe that these results may have significant implications for policy formulation and resource allocation with critical societal and economic impacts.

\section*{Acknowledgements}
This research is partly funded by NSF grants SaTC-2133169, RI-2232054, and CAREER-2143706. The views and conclusions of this work are those of the authors only. The authors are also thankful to Christine Task for early discussion and feedback regarding this topic.


\newpage
\appendix
\label{sec:appendix}

\onecolumn
\begin{center}
\LARGE{\bf Appendix}
\end{center}

\addtocontents{toc}{\protect\setcounter{tocdepth}{2}}
\tableofcontents


\FloatBarrier
\section{Experimental Details}
\label{app:experimental-details}

\paragraph{Experimental settings.} Each of the results in this paper was produced using 1 CPU (Intel) with 64 GB of memory per task. We conducted all experiments using ACS data from IPUMS \cite{IPUMS} for 2021 and 2022. Following standard survey design practices, 2021 data was used as prior information, while 2022 data served as the ground truth from which samples were drawn.

\paragraph{Data Preprocessing.}
To ensure meaningful analysis, we preprocessed the ACS data as follows:
\begin{enumerate}
\item \textbf{Handling Missing Values:} Records with unknown income or educational levels were removed.
\item \textbf{Merging Sparse Racial Groups:} To prevent issues with extremely small sample sizes, sparse racial groups were consolidated. For example, Asian subgroups (e.g., Chinese, Japanese, Pacific Islanders) were combined into a single "Asian" category, and categories like "two major races" and "three or more major races" were merged into "2 or more races."
\item \textbf{Merging Sparse Educational Levels:} Similar to racial groups, educational levels with small sample sizes were combined. For instance, "Grade 9," "Grade 10," and "Grade 11" were merged into a group that consists of high school educated people who did not graduate, and "No schooling" and "Nursery to grade 4" were combined into a group consisting of people with limited education. Details are provided in Table~\ref{table:education_level}.
\item \textbf{Cloning Weighted Records:} Each record in the ACS data, representing multiple individuals via the {\texttt{PERWT}} attribute, was cloned by its weight to better reflect population distribution in our analysis.
\end{enumerate}

\paragraph{Geographical divisions.}
The smallest geographical unit provided in ACS data is Public Use Microdata Area (\texttt{PUMA}). Each \texttt{PUMA} contains at least 100,000 people in order to maintain confidentiality. However, due to the nature of large group size, \texttt{PUMA} is not a great representation of the geographical unit to be used during the second phase (door-to-door) of the survey. Therefore, we artificially divided them into smaller regions of approximately 4,000 people, emulating Census Tract data. Although actual tract sizes vary, 4,000 is the optimal size recommended by the Census Bureau.

\FloatBarrier
\section{Additional Experiments on Other States}
\label{app:other_states}

This section presents results for Maine and Nevada, chosen for their contrasting levels of racial diversity—Maine being more homogeneous and Nevada more diverse. These states were selected based on the Census Bureau's Simpson's Index of Diversity \cite{Census_Diversity_Index}, calculated as:
\[ D = 1 - \sum_{i \in [G]}\left(\frac{N_i}{N}\right)^2\]
Here, $D$ ranges from 0 to 1, representing the probability that two randomly chosen individuals will belong to different racial groups. For detailed population sizes and diversity scores, refer to Table~\ref{tab:state_characteristics}.

\begin{table}[h!]
\centering
\resizebox{0.6 \columnwidth}{!}{
    \begin{tabular}{|c|c|c|c|c|}
    \hline
    State          & Population  & Racial Diversity Index & Income Per Capita (\$) \\ \hline
    Maine          & 1,170,363   & 0.16 (Low Diversity)   & 46,087                 \\ \hline
    Connecticut    & 3,017,448   & 0.52 (Baseline)        & 59,772                 \\ \hline
    Nevada         & 2,566,612   & 0.67 (High Diversity)  & 44,742                 \\ \hline
    \end{tabular}}
\caption{Characteristics of three states after data pre-processing done on ACS 2022 dataset.}
\label{tab:state_characteristics}
\end{table}

\paragraph{Proxy plots.}
Figure~\ref{app:fig:proxy_NV} and Figure~\ref{app:fig:proxy_ME} display the proxy plots for Nevada and Maine, respectively. Note that the estimator's variance is \emph{high} in Maine for most racial groups, due to the small population sizes of minorities. 

\begin{figure*}[t]
    \centering
    \begin{subfigure}[t]{0.195\textwidth}
        \centering
        \includegraphics[width=\linewidth]{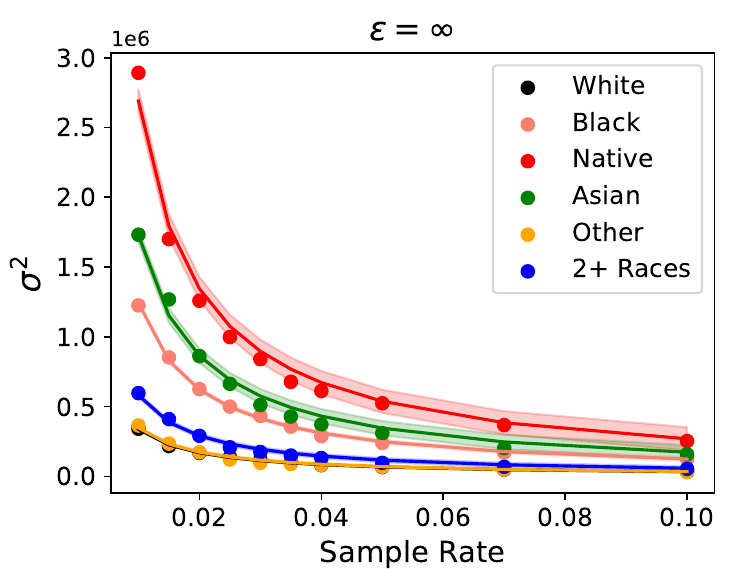}
    \end{subfigure}
    \begin{subfigure}[t]{0.195\textwidth}
        \centering
        \includegraphics[width=\linewidth]{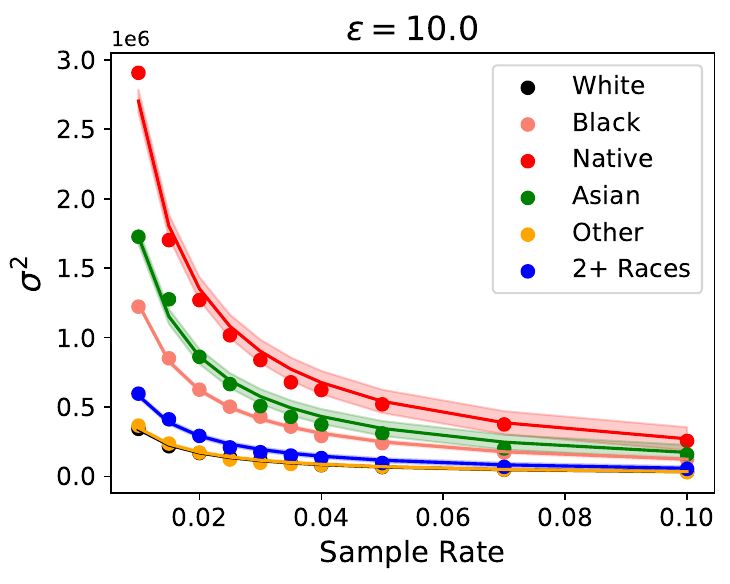}
    \end{subfigure}
    \begin{subfigure}[t]{0.195\textwidth}
        \centering
        \includegraphics[width=\linewidth]{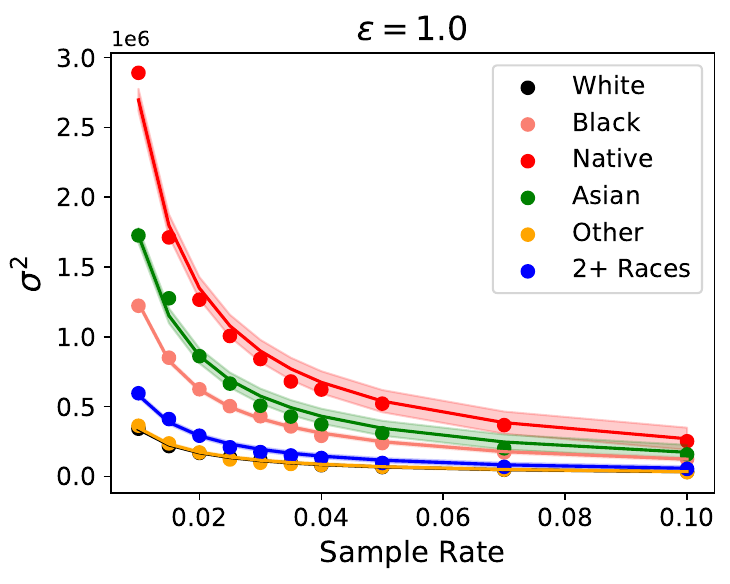}
    \end{subfigure}
    \begin{subfigure}[t]{0.195\textwidth}
        \centering
        \includegraphics[width=\linewidth]{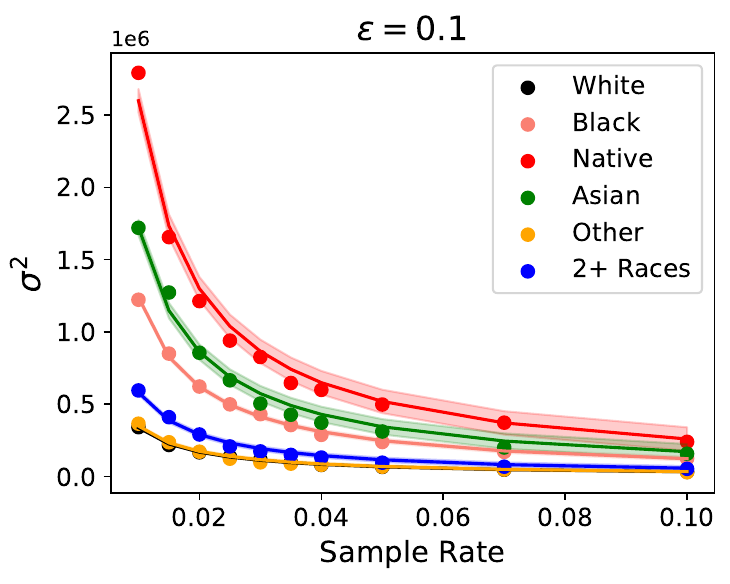}
    \end{subfigure}
    \begin{subfigure}[t]{0.195\textwidth}
        \centering
        \includegraphics[width=\linewidth]{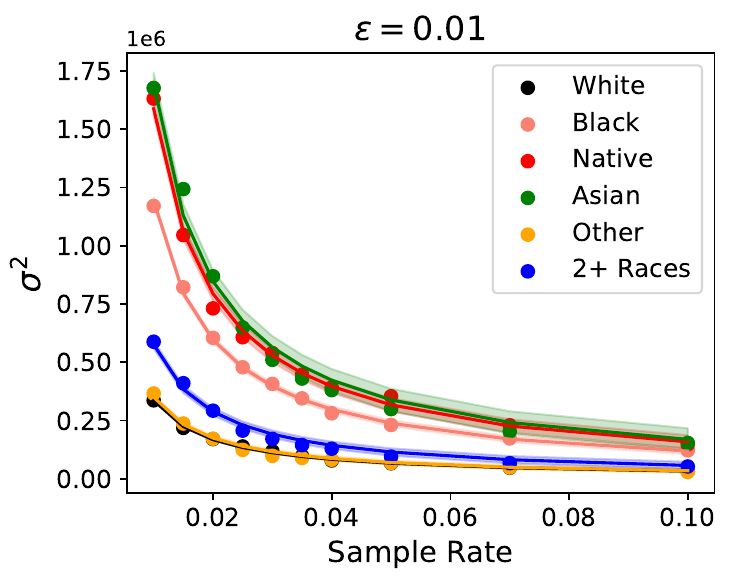}
    \end{subfigure}
    \caption{Estimating the variance of mean \emph{income} in Nevada using \emph{race} as a subgroup with different privacy budget $\varepsilon$. Results averaged over 200 trials and 200 data points.}
    \label{app:fig:proxy_NV}
\end{figure*}

\begin{figure*}[h]
    \centering
    \begin{subfigure}[t]{0.195\textwidth}
        \centering
        \includegraphics[width=\linewidth]{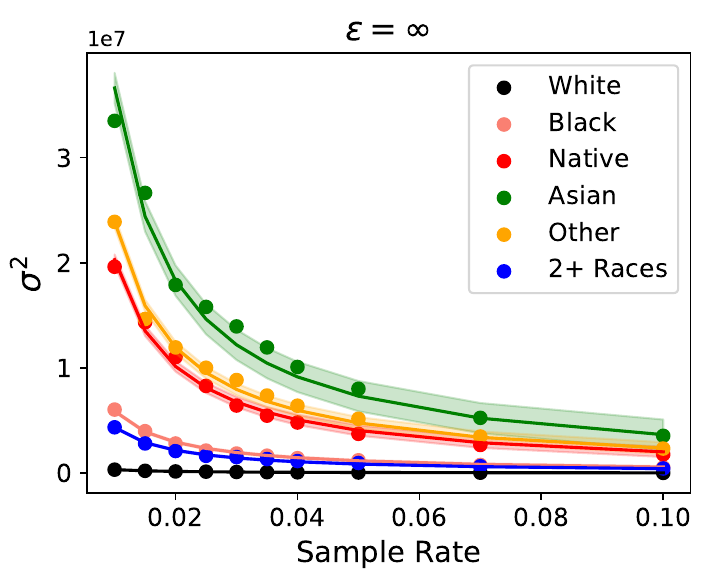}
    \end{subfigure}
    \begin{subfigure}[t]{0.195\textwidth}
        \centering
        \includegraphics[width=\linewidth]{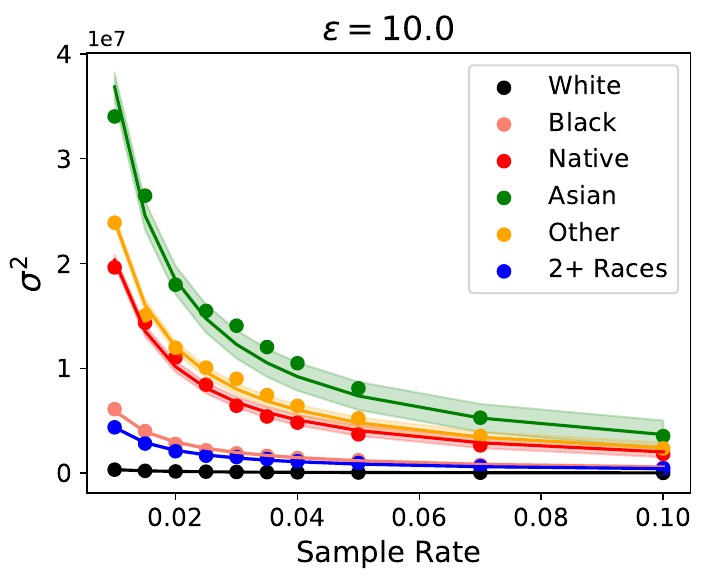}
    \end{subfigure}
    \begin{subfigure}[t]{0.195\textwidth}
        \centering
        \includegraphics[width=\linewidth]{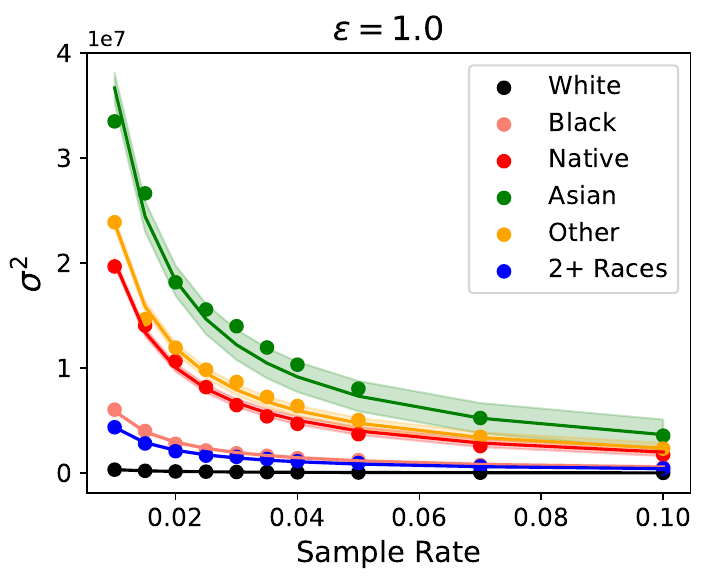}
    \end{subfigure}
    \begin{subfigure}[t]{0.195\textwidth}
        \centering
        \includegraphics[width=\linewidth]{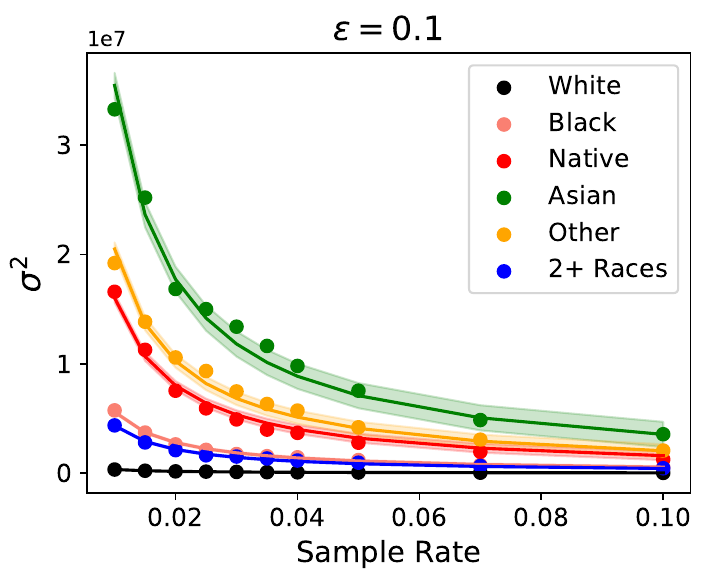}
    \end{subfigure}
    \begin{subfigure}[t]{0.195\textwidth}
        \centering
        \includegraphics[width=\linewidth]{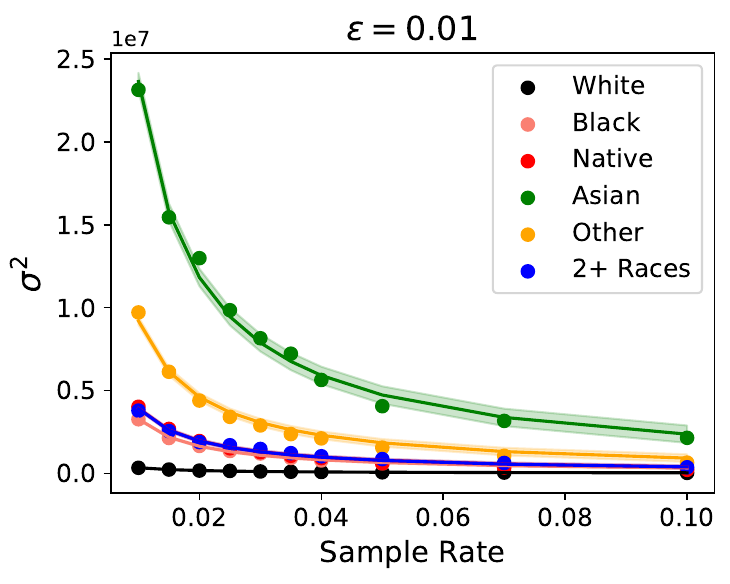}
    \end{subfigure}
    \caption{Estimating the variance of mean \emph{income} in Maine using \emph{race} as a subgroup with different privacy budget $\varepsilon$. Results averaged over 200 trials and 200 data points.}
    \label{app:fig:proxy_ME}
\end{figure*}

\begin{figure}[h]
    \centering
    \begin{subfigure}[t]{0.48\textwidth}
        \centering
        \includegraphics[width=\linewidth]{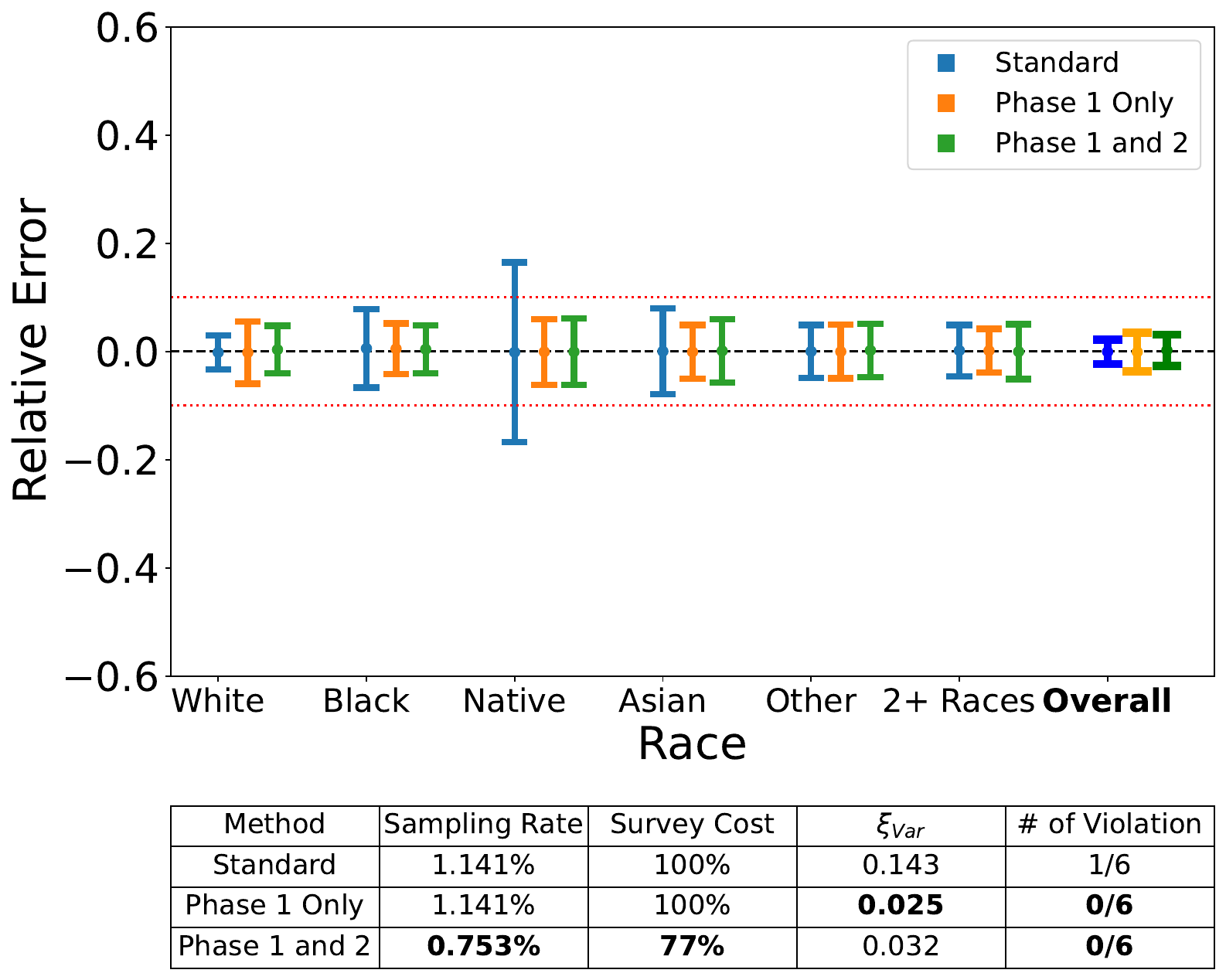}
        \caption{\emph{Nevada}}
        \label{fig:rel_errors_NV_race}
    \end{subfigure}
    \hfill
    \begin{subfigure}[t]{0.48\textwidth}
        \centering
        \includegraphics[width=\linewidth]{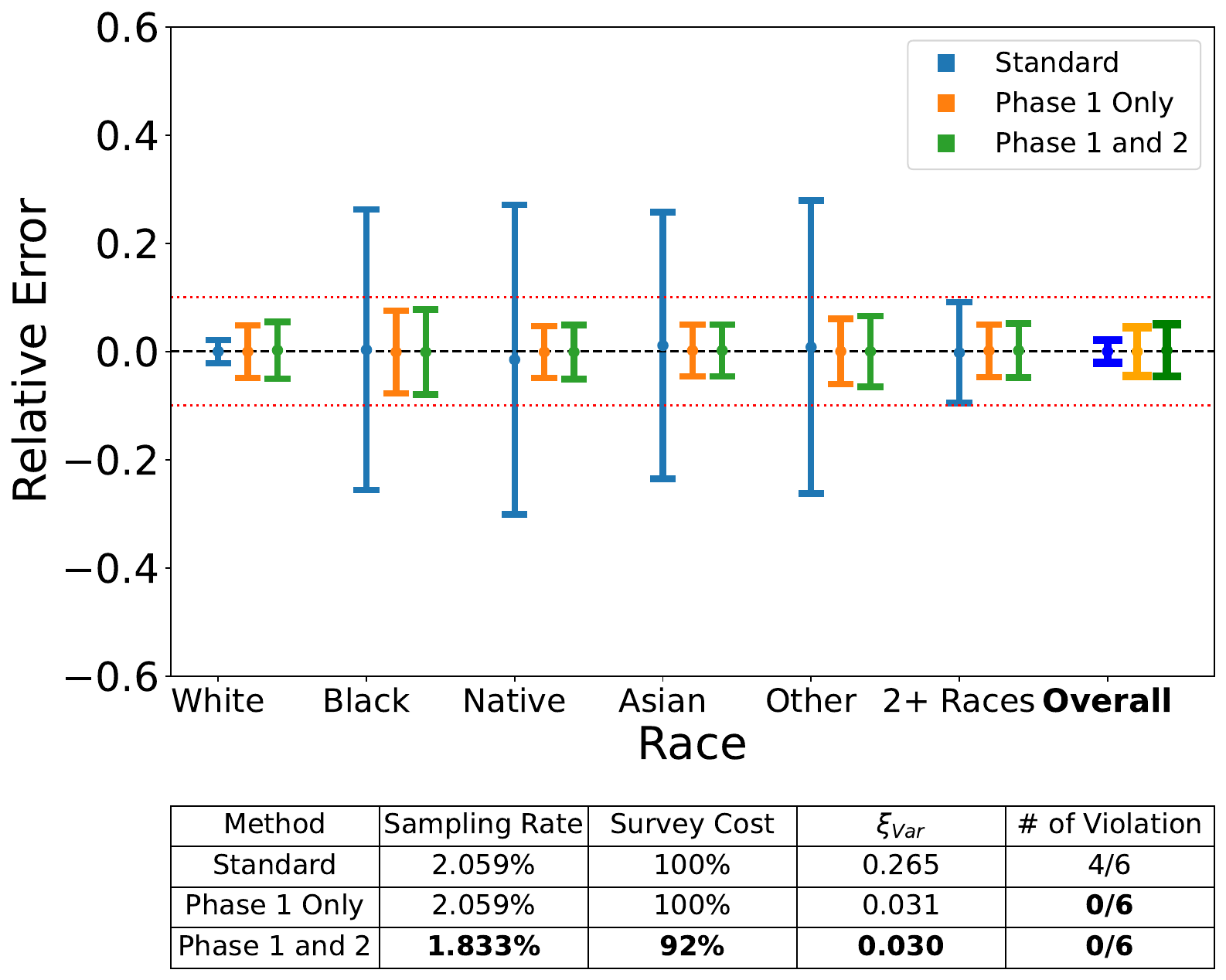}
        \caption{\emph{Maine}}
        \label{fig:rel_errors_ME_race}
    \end{subfigure}
    \caption{Relative errors from estimating average \emph{income} for each \emph{race} in Nevada and Maine}
    \label{fig:rel_errors_other_states}
\end{figure}

\paragraph{Optimized Sampling.}
The \emph{Standard Allocation} method performs much better in Nevada, as shown in Figure~\ref{fig:rel_errors_NV_race}, compared to Maine in Figure~\ref{fig:rel_errors_ME_race}. This is because Nevada's more homogeneous population allows for proportional sampling to allocate surveys more equitably across racial groups. In contrast, Maine's smaller minority population makes this approach much less effective.

\paragraph{DP sampling.}
The surprising finding from Section~\ref{sssec:DP-sampling}—that adding more noise can reduce the variance of errors—is evident in both states, as shown in Figures~\ref{fig:dp_plots_NV} and \ref{fig:dp_plots_ME}. This effect is more pronounced in Maine due to its smaller racial minority population, leading to a greater positive bias, as discussed in Section~\ref{sec:sampling_privacy}. Table~\ref{table:population_estimation_ME} and Table~\ref{table:population_estimation_NV} show the impact of differential privacy on the estimated population sizes for Maine and Nevada, respectively.

\begin{figure*}[h]
    \centering
    \begin{subfigure}[t]{1.0\textwidth}
        \centering
        \includegraphics[width=1.0\textwidth]{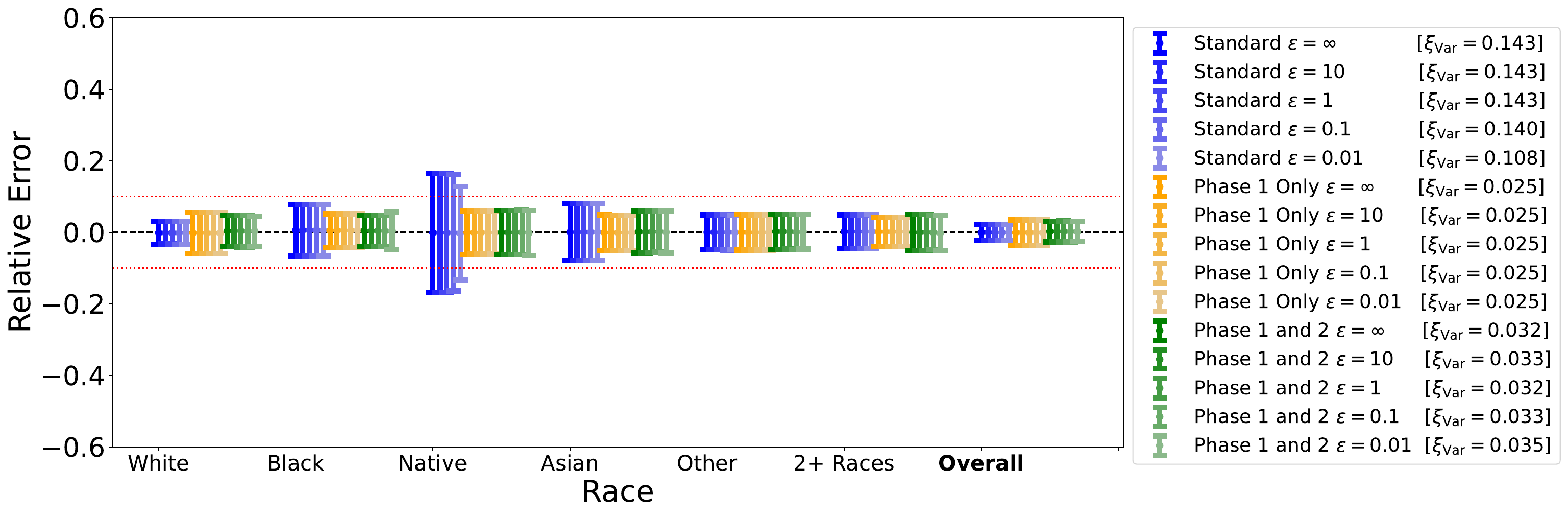}
        \caption{\emph{Nevada}}
        \label{fig:dp_plots_NV}
    \end{subfigure}
    \hfill
    \begin{subfigure}[t]{1.0\textwidth}
        \centering
        \includegraphics[width=1.0\textwidth]{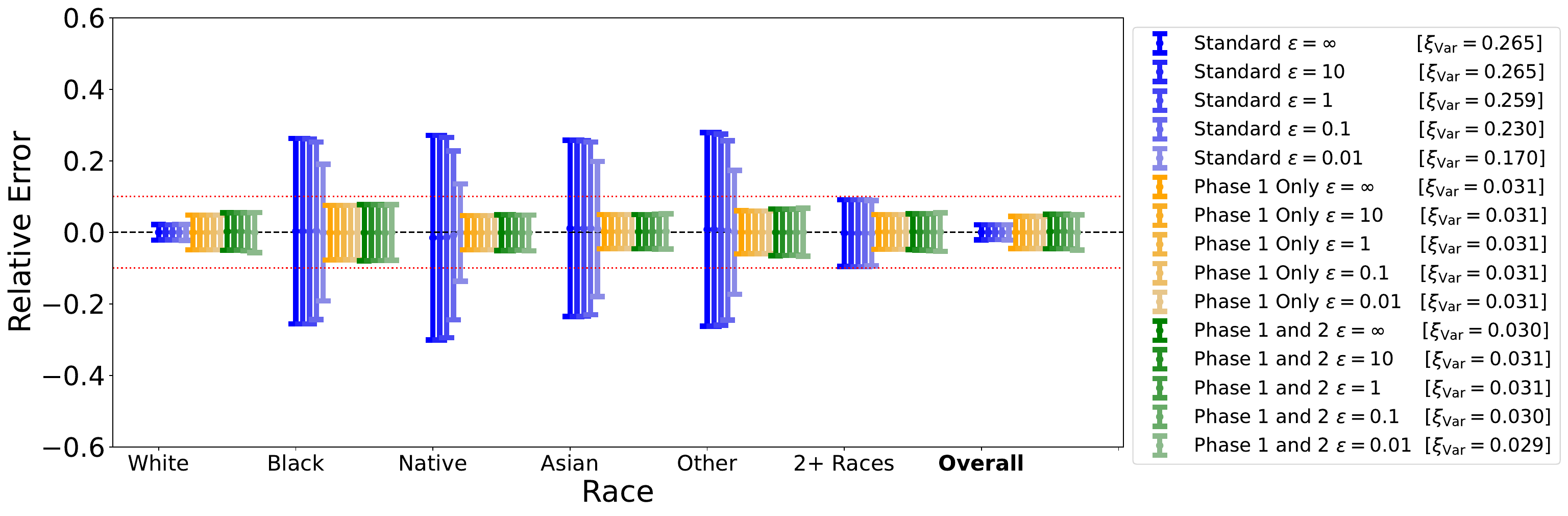}
        \caption{\emph{Maine}}
        \label{fig:dp_plots_ME}
    \end{subfigure}
    \caption{Relative errors from estimating mean income using DP-noised $N_i^r$ for Nevada and Maine. Each region used in the phase 2 contains approximately 4,000 people, similar to the size of Census Tracts.}
    \label{fig:dp_plots_NV_ME}
\end{figure*}

\begin{table}[h]
\centering
\resizebox{0.7\columnwidth}{!}{
    \begin{tabular}{|c|c|c|c|c|c|c|c|}
    \hline
    $\varepsilon$ \symbol{92} Race & White      & Black   & Native  & Asian   & Other   & 2+ Races  & Total        \\ \hline
    $\infty$                       & 1,065,951  & 15,952  & 3,695   & 15,454  & 8,554   & 60,757    & 1,170,363    \\ \hline
    $10$                           & 1,065,797  & 15,873  & 3,678   & 15,362  & 8,507   & 60,623    & 1,169,840    \\ \hline
    $1$                            & 1,065,740  & 15,917  & 3,741   & 15,412  & 8,596   & 60,636    & 1,170,042    \\ \hline
    $0.1$                          & 1,065,167  & 16,795  & 4,855   & 16,094  & 9,897   & 60,889    & 1,173,697    \\ \hline
    $0.01$                         & 1,059,408  & 27,982  & 16,641  & 25,435  & 24,194  & 66,630    & 1,220,290    \\ \hline
    \end{tabular}}
\caption{Estimated population for each race in Maine using prior (e.g. ACS 2021 dataset).}
\label{table:population_estimation_ME}
\end{table}

\begin{table}[h]
\centering
\resizebox{0.7\columnwidth}{!}{
    \begin{tabular}{|c|c|c|c|c|c|c|c|}
    \hline
    $\varepsilon$ \symbol{92} Race & White      & Black   & Native  & Asian   & Other   & 2+ Races  & Total        \\ \hline
    $\infty$                       & 1,337,725  & 226,955 & 37,721  & 251,807 & 302,774 & 409,630   & 2,566,612    \\ \hline
    $10$                           & 1,337,394  & 226,632 & 37,546  & 251,472 & 302,445 & 409,314   & 2,564,803    \\ \hline
    $1$                            & 1,337,330  & 226,646 & 37,618  & 251,468 & 302,408 & 409,322   & 2,564,792   \\ \hline
    $0.1$                          & 1,336,724  & 226,897 & 39,084  & 251,368 & 301,904 & 409,376   & 2,565,353  \\ \hline
    $0.01$                         & 1,330,540  & 234,827 & 60,706  & 255,013 & 301,702 & 412,266   & 2,595,054   \\ \hline
    \end{tabular}}
\caption{Estimated population for each race in Nevada using prior (e.g. ACS 2021 dataset).}
\label{table:population_estimation_NV}
\end{table}


\FloatBarrier
\section{Additional Experiments on Education Level}
\label{app:education_level}

This section presents results using \emph{educational level} as a subgroup instead of \emph{race}. For details on each educational level, see Table~\ref{table:education_level}. The following experiments use Connecticut as the state of focus.

\begin{table}[h]
\centering
\begin{tabular}{|c|c|}
\hline
Education & Description \\ \hline
1   & No schooling to grade 4   \\ \hline
2   & Grade 5, 6, 7, or 8   \\ \hline
5   & Grade 9, 10, or 11   \\ \hline
6   & Grade 12   \\ \hline
7   & 1 year of college   \\ \hline
8   & 2 years of college   \\ \hline
10  & 4 years of college   \\ \hline
11  & 5+ years of college   \\ \hline
\end{tabular}
\caption{Description for each education level}
\label{table:education_level}
\end{table}

\paragraph{Proxy plots.} Figure~\ref{fig:proxy_EDUC} presents the \emph{proxy plot}, as discussed in Section \ref{sec:empirical_variance_estimation}, for Connecticut, using \emph{education level} as the subgroup.

\begin{figure*}[h]
    \centering
    \begin{subfigure}[t]{0.195\textwidth}
        \centering
        \includegraphics[width=\linewidth]{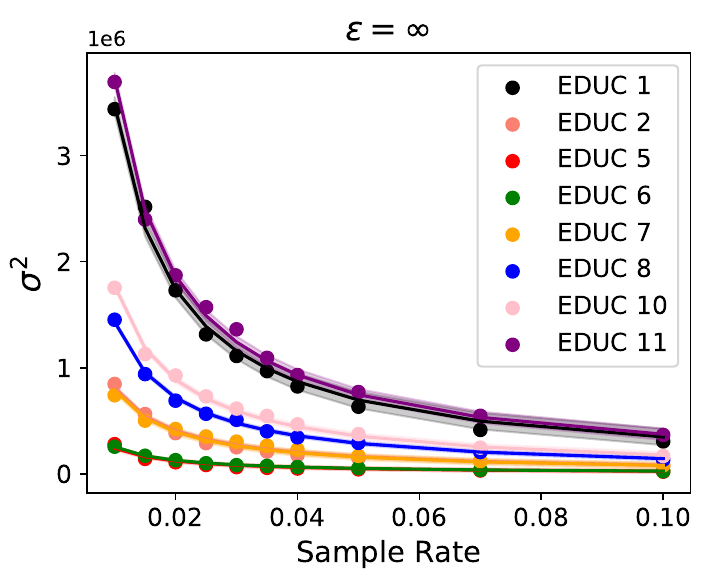}
    \end{subfigure}
    \begin{subfigure}[t]{0.195\textwidth}
        \centering
        \includegraphics[width=\linewidth]{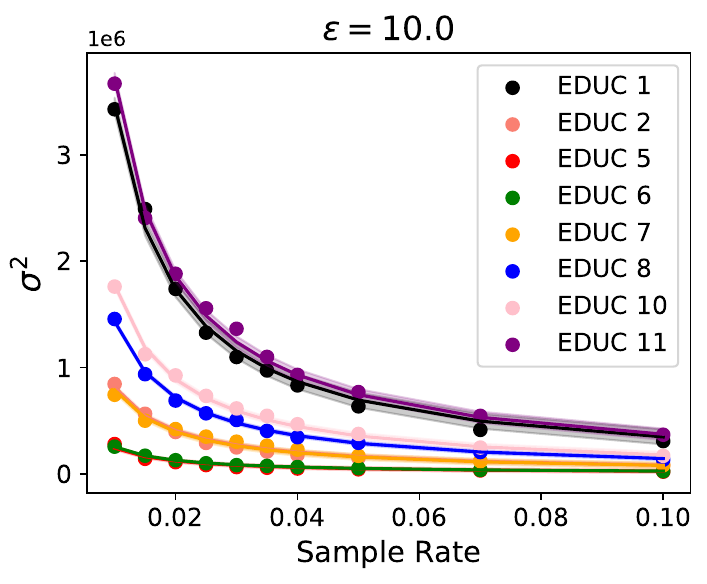}
    \end{subfigure}
    \begin{subfigure}[t]{0.195\textwidth}
        \centering
        \includegraphics[width=\linewidth]{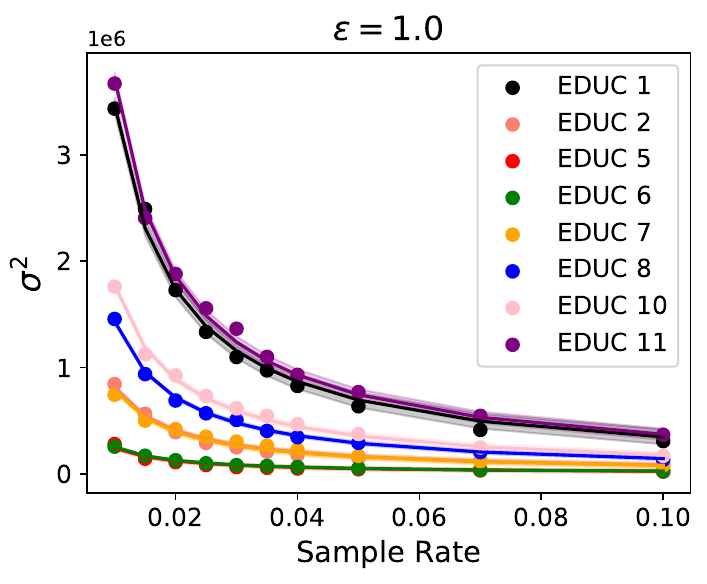}
    \end{subfigure}
    \begin{subfigure}[t]{0.195\textwidth}
        \centering
        \includegraphics[width=\linewidth]{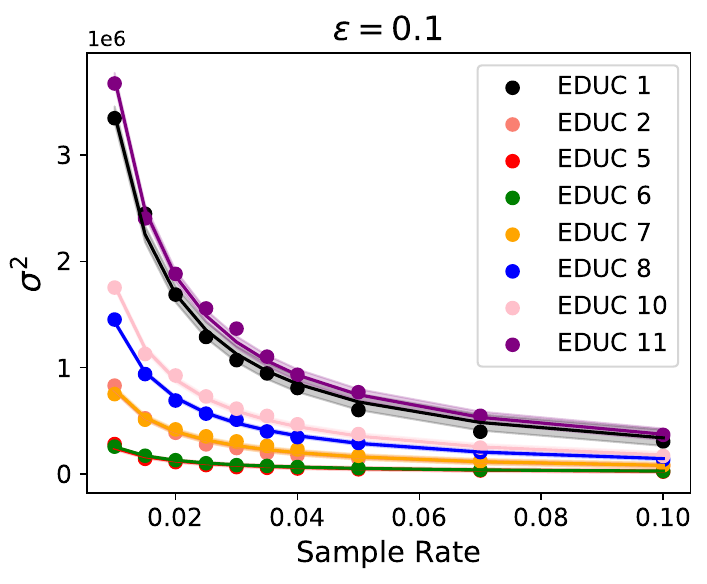}
    \end{subfigure}
    \begin{subfigure}[t]{0.195\textwidth}
        \centering
        \includegraphics[width=\linewidth]{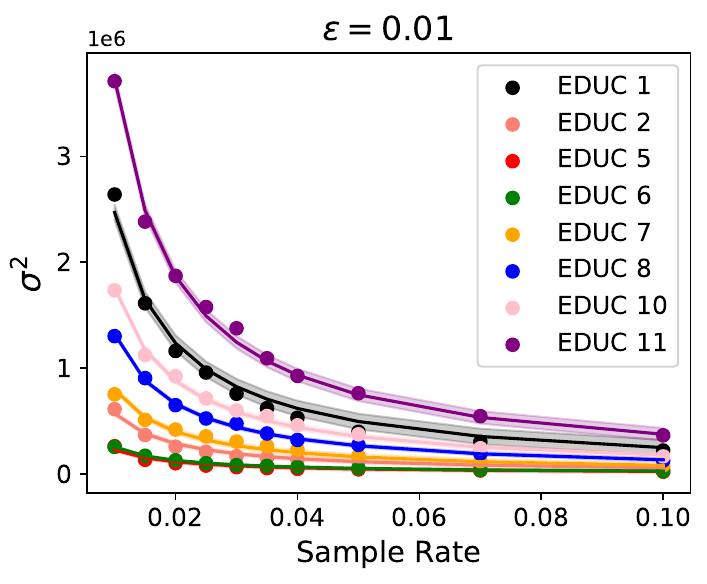}
    \end{subfigure}
    \caption{Estimating the variance of mean \emph{income} in Connecticut using \emph{education level} as a subgroup with different privacy budget $\varepsilon$. Results averaged over 200 trials and 200 data points.}
    \label{fig:proxy_EDUC}
\end{figure*}

\paragraph{Optimized Sampling.} As shown in Figure~\ref{fig:rel_err_educ}, the \emph{Standard Allocation} method fails to meet confidence constraints for all groups with the highest education level below high school (1, 2, and 5). Both optimization approaches effectively reduce these error disparities.

\begin{figure}[h]
    \centering
    \includegraphics[width=0.45\columnwidth]{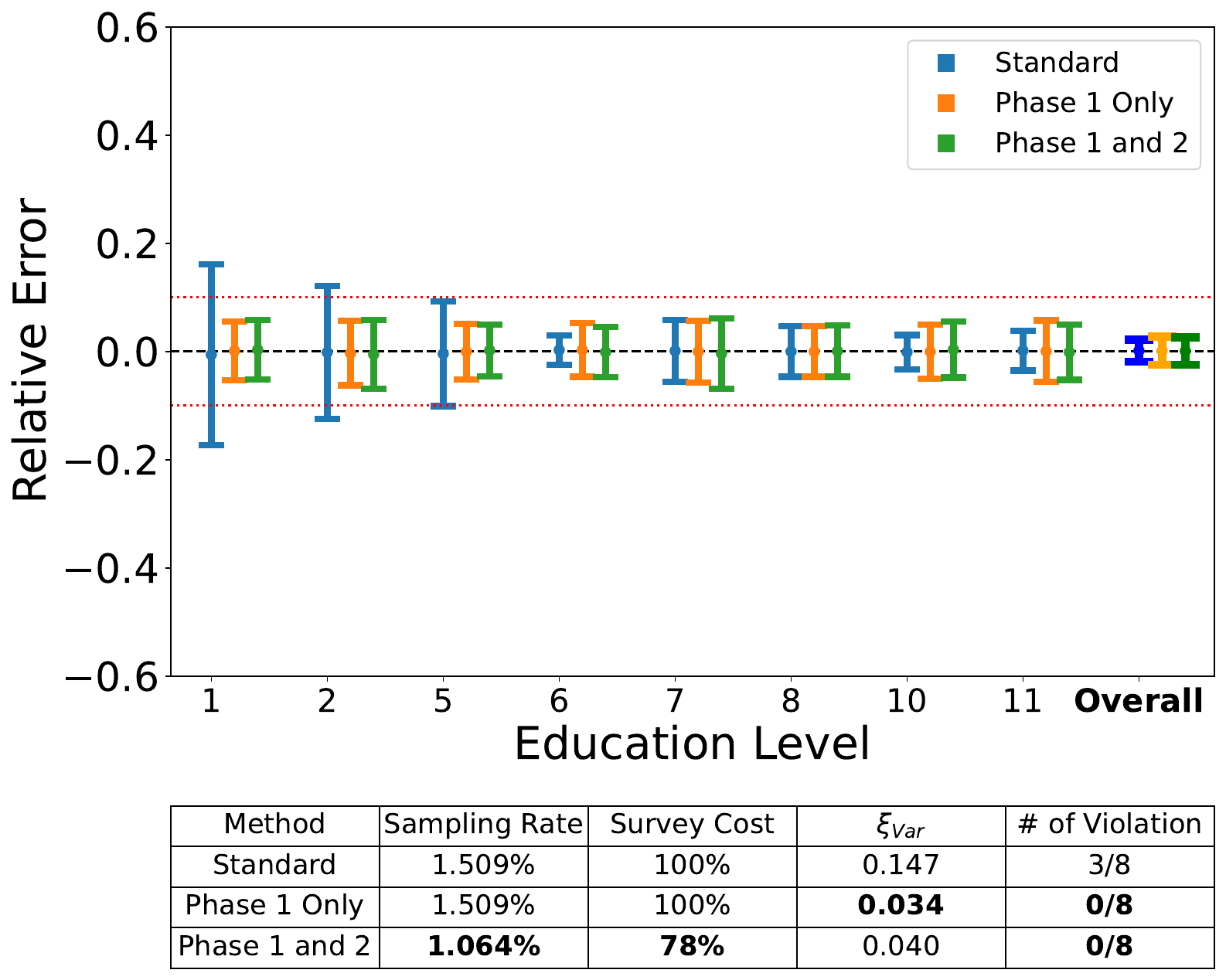}
    \caption{Relative errors from estimating average \emph{income} in Connecticut for each \emph{education level}}
    \label{fig:rel_err_educ}
\end{figure}

\paragraph{DP sampling.}  The surprising finding from Section~\ref{sssec:DP-sampling}—that adding more noise can reduce the variance of errors—can also be observed in Figure~\ref{fig:dp_educ}.

\begin{figure*}[h]
    \centering
    \includegraphics[width=0.8\textwidth]{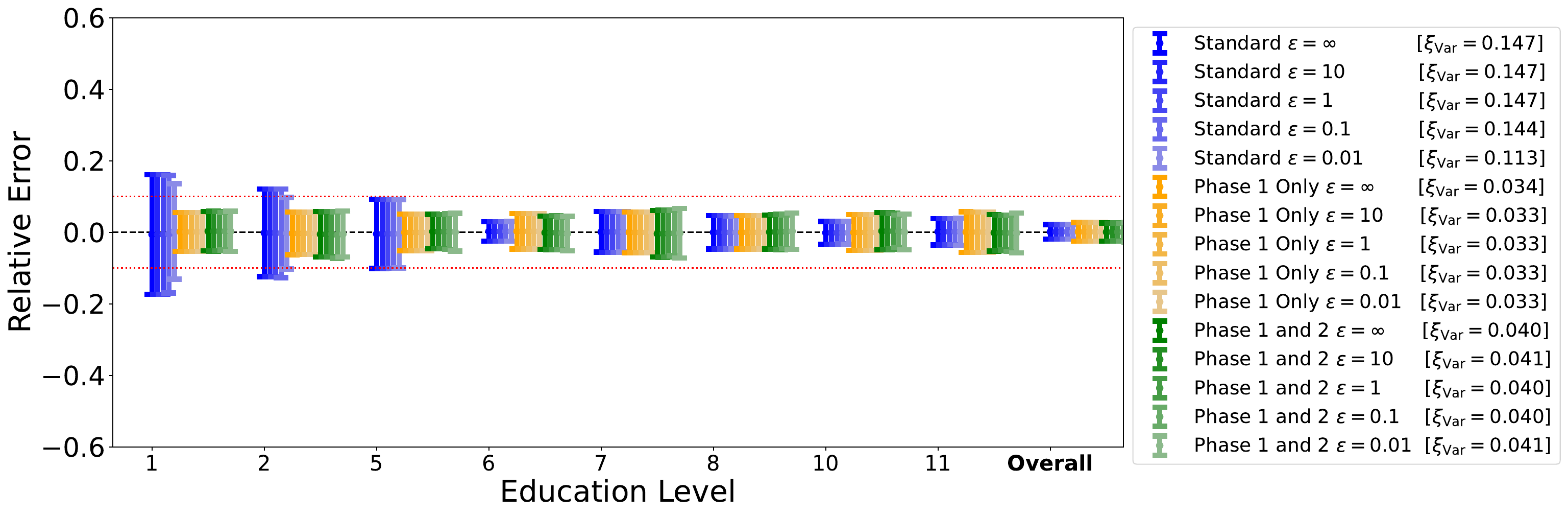}
    \caption{Relative errors in estimating mean income using \emph{education level} as a subgroup, with DP-noised $N_i^r$, are shown for Connecticut. Each region used in the phase 2 contains approximately 4,000 people, similar to the size of Census Tracts.}
    \label{fig:dp_educ}
\end{figure*}

\FloatBarrier
\section{Phase 1 Algorithm}
\label{sec:algorithm}

The algorithm for solving the \emph{Phase 1 Only} optimization approach involves solving the following program.
    {\large
\begin{subequations}
    \label{eq:6}
    \begin{align}
        \minimize_{\vcol{\bm p}} &\; 
         \underbrace{c_1 \left(\sum_{i \in [G]} \vcol{p_i} N_i \right)}_{1^{\text{st}} \text{ phase cost}} \label{obj:6a}\\
    \hspace{-10pt}
            \texttt{s.t.}~~&
                {n}_i = \underbrace{\vcol{p_i} N_i \left(1-{F_i^1}\right)}_{1^{\text{st}} \text{ phase samples}}
                         \; \forall i \in [G] \label{c:6b}\\
            &\; \Pr(|\err(\hat{\bm \theta}_i({n}_i))| > \gamma_i) \leq \alpha, \;\; \forall i \in [G], \label{c:6c}\\
            &\; 0 \leq \vcol{p_i} \leq 1, \;\; \forall i \in [G] \label{c:6d}
\end{align}
\end{subequations}
}

\FloatBarrier
\section{Varying Parameters}
\label{sec:varying_param}
\begin{figure}[h]
    \centering
    \begin{subfigure}[h]{0.3\columnwidth}
        \centering
        \includegraphics[width=\linewidth]{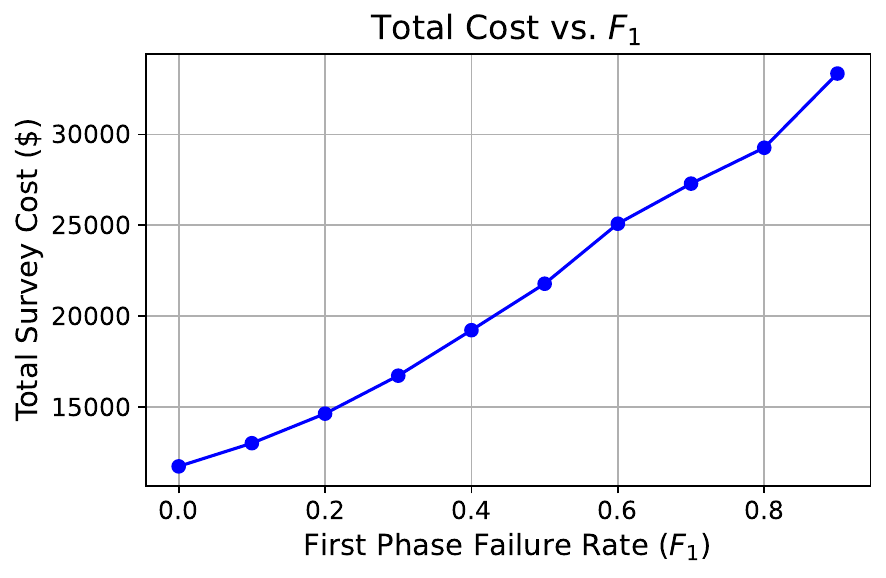}
        \caption{Effect of $F_1$ on the total cost.}
        \label{fig:varying_param_F1}
    \end{subfigure}
    \begin{subfigure}[h]{0.3\columnwidth}
        \centering
        \includegraphics[width=\linewidth]{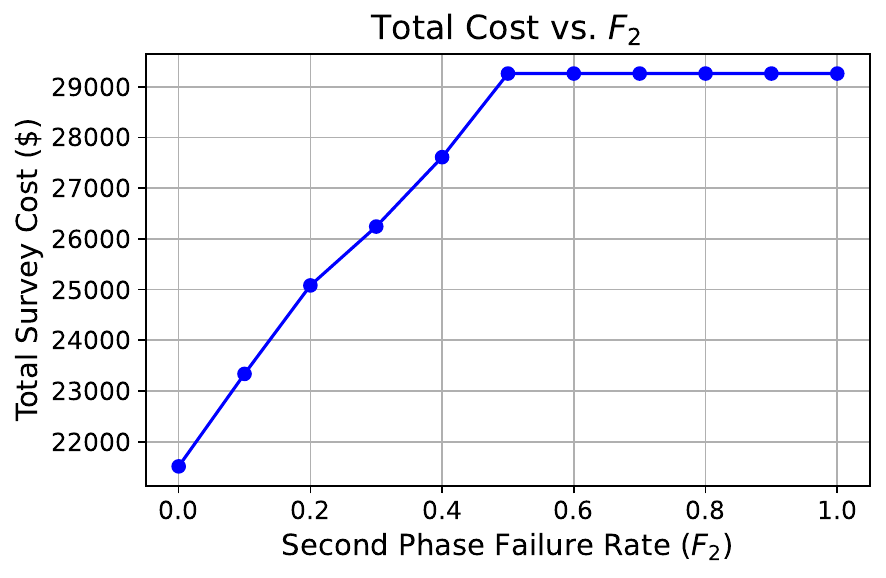}
        \caption{Effect of $F_2$ on the total cost.}
        \label{fig:varying_param_F2}
    \end{subfigure}
    \vskip\baselineskip
    \begin{subfigure}[h]{0.3\columnwidth}
        \centering
        \includegraphics[width=\linewidth]{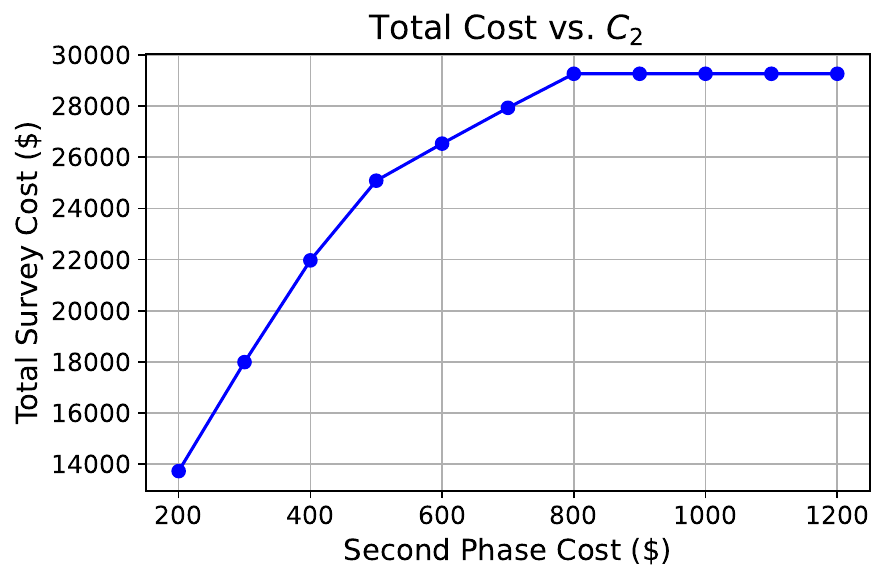}
        \caption{Effect of $c_2$ on the total cost.}
        \label{fig:varying_param_c2}
    \end{subfigure}
    \begin{subfigure}[h]{0.3\columnwidth}
        \centering
        \includegraphics[width=\linewidth]{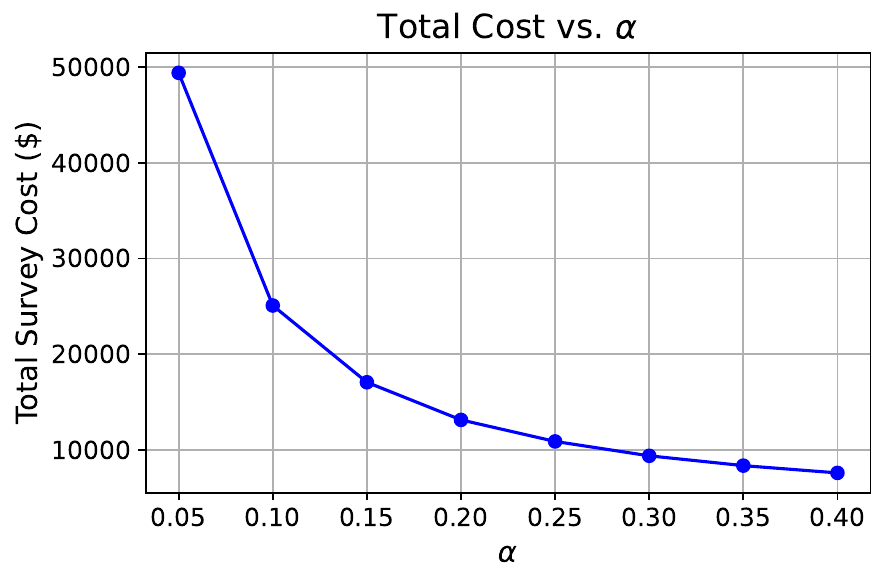}
        \caption{Effect of $\alpha$ on the total cost.}
        \label{fig:varying_param_alpha}
    \end{subfigure}
    \begin{subfigure}[h]{0.3\columnwidth}
        \centering
        \includegraphics[width=\linewidth]{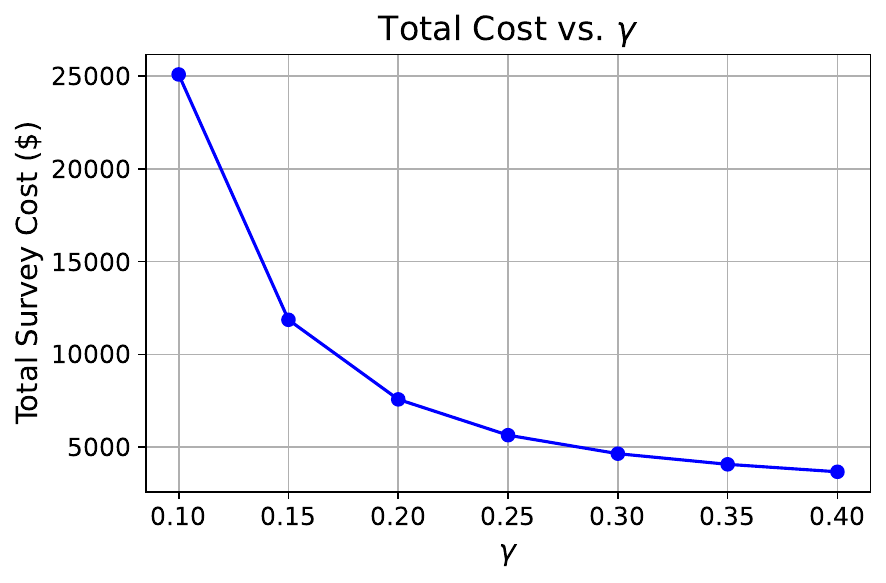}
        \caption{Effect of $\gamma$ on the total cost.}
        \label{fig:varying_param_gamma}
    \end{subfigure}
    \caption{Effect of varying parameters on the total cost of the survey. The study was done on estimating average \emph{income} in \emph{Connecticut} for each \emph{race}. The base parameters were $F_1 = 0.6, F_2 = 0.2, c_1 = 1, c_2 = 500, \alpha = 0.1, \text{and } \gamma = 0.1$. }
    \label{fig:varying_param}
\end{figure}

\textbf{Ablation study for total costs: }
\begin{itemize}
    \item \textbf{Total Cost vs. $F_1$:} The increasing cost trend is pretty clear. This is a direct cause from needing to sample more to meet the required number of \emph{successful} samples.
    \item \textbf{Total Cost vs. $F_2$:} The increasing cost trend is pretty clear. Notice that the total cost curve plateaus near $0.6$, which is the base value of $F_1$. When $F_2 \geq F_1$, there is no merit to utilizing the second-phase at all as it is more costly per-survey wise. Therefore, the optimizer selects only from the first-phase to minimize the cost once $F_2 \geq F_1$.
    \item \textbf{Total Cost vs. $c_2$:} The increasing cost trend is pretty clear. Notice the total cost curve starts to plateau as $c_2$ increases. This is due to the optimizer preferring to sample more from the first-phase once $c_2$ becomes too expensive (at which point, an equal or better level of utility can simply be achieved by allocating more surveys during the first phase only for less cost).
    \item \textbf{Total Cost vs. $\alpha$:} The trend is downward as relaxing $\alpha$ allows for a bigger margin for error due to a lower confidence interval, which leads to less sampling.
    \item \textbf{Total Cost vs. $\gamma$:} The trend is downward as relaxing $\gamma$ allows for a bigger margin for error under $1-\alpha$ confidence interval, which leads to less sampling.
\end{itemize}

\FloatBarrier
\section{Sparsity Analysis}
\label{app:sparsity_analysis}
In Section~\ref{sec:sampling_privacy}, we discussed how the size of the region influences the magnitude of bias when using a fixed privacy budget  (Corollary~\ref{cor:bias_aggr}). Figure~\ref{fig:sparsity} presents the results of the same experiment shown in Figure~\ref{fig:relative_error_with_privacy}, but across three different levels of sparsity: 200,000 People Per Region (PPR), 4,000 PPR, and 2,000 PPR. These levels correspond to sampling regions at the county level, Census Tract level, and Census Block Group level, respectively. As we reduce the size of the regions, the variance of the error decreases under the \emph{Standard Allocation} method, even though the privacy budget $\varepsilon$ remains constant. This decrease in variance is analogous to the effect observed with differential privacy sampling when the privacy budget increases, except here the reduction in bias is due to a smaller $N_i^r$ (the number of individuals in each region), rather than a smaller $\varepsilon$.

\begin{figure*}[t]
    \centering
    \begin{subfigure}[h]{1.00\textwidth}
        \centering
        \includegraphics[width=\linewidth]{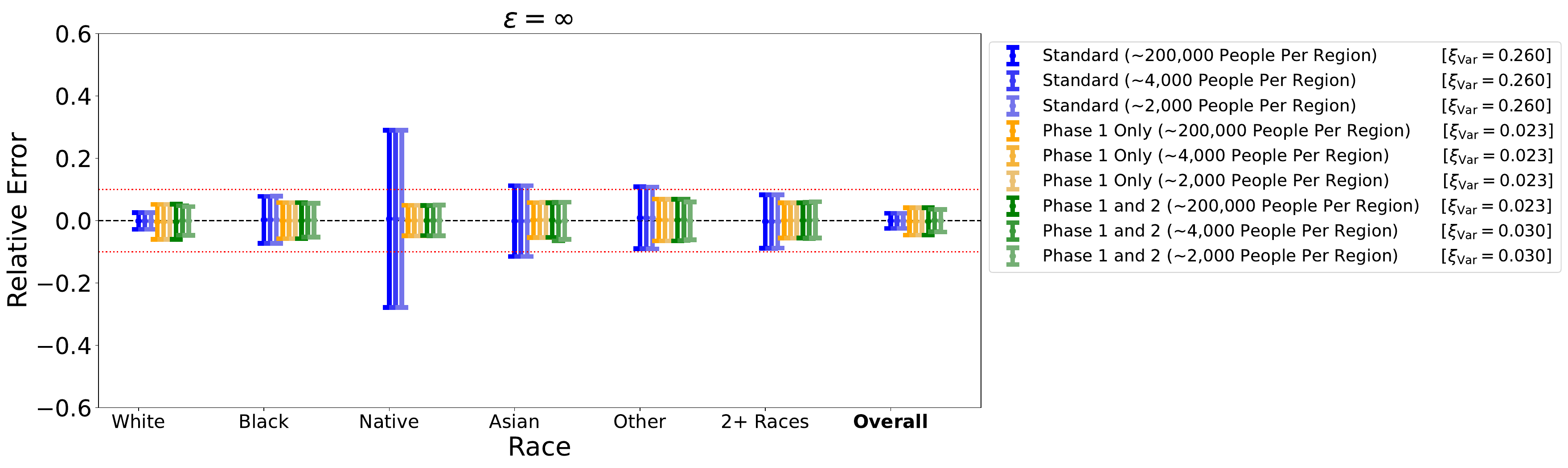}
    \end{subfigure}
    \hfill
    \begin{subfigure}[h]{1.00\textwidth}
        \centering
        \includegraphics[width=\linewidth]{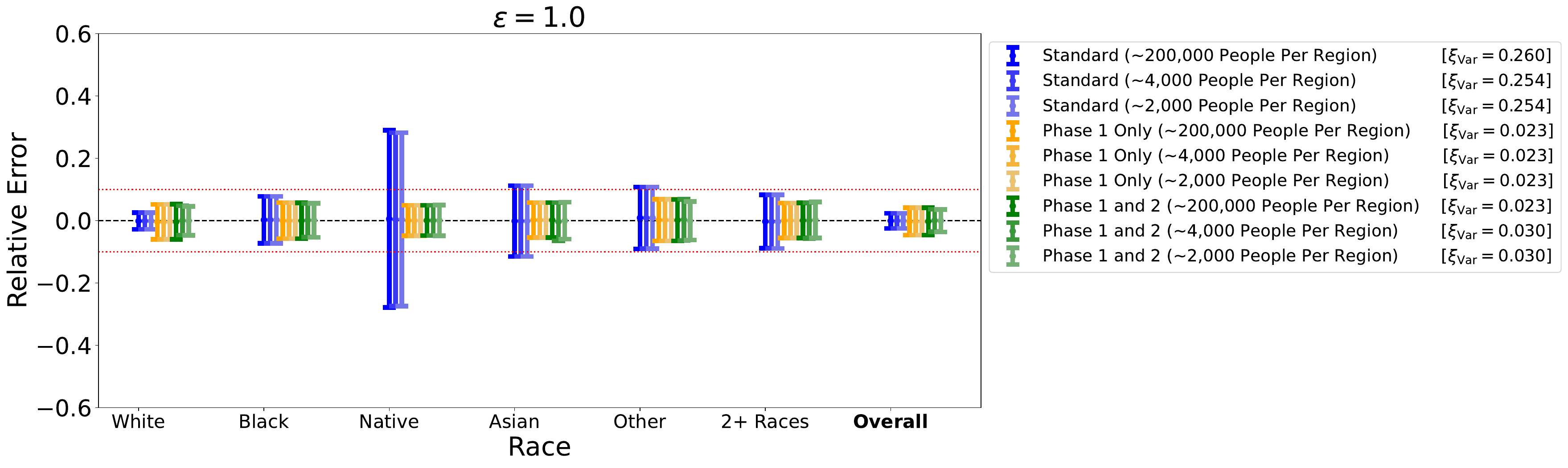}
    \end{subfigure}
    \hfill
    \begin{subfigure}[h]{1.00\textwidth}
        \centering
        \includegraphics[width=\linewidth]{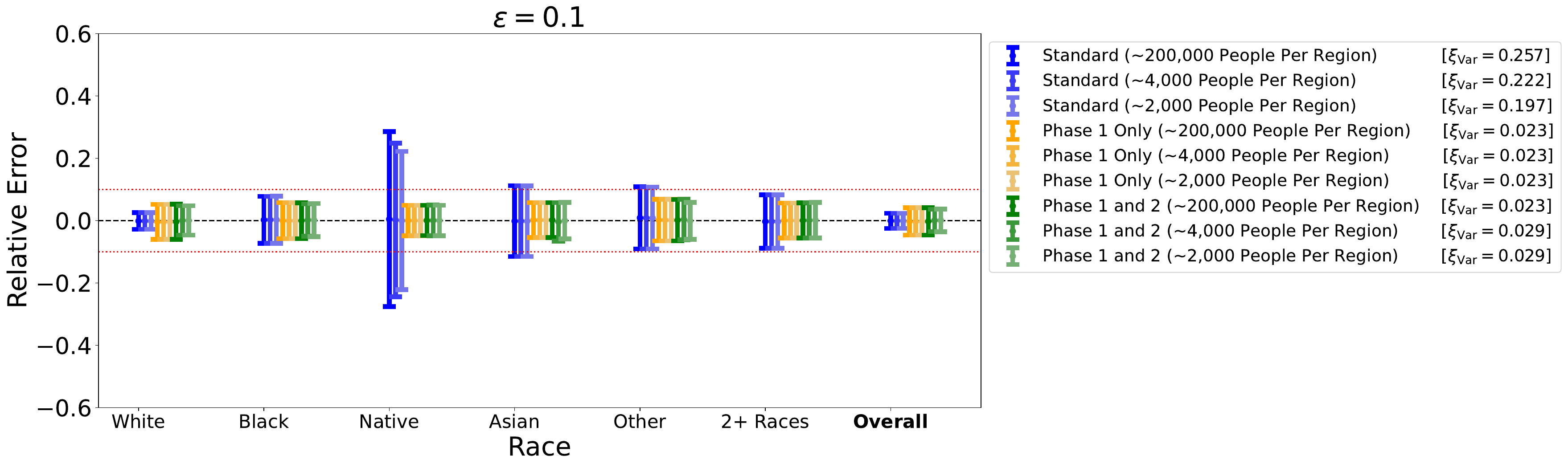}
    \end{subfigure}    
    \hfill
    \begin{subfigure}[h]{1.00\textwidth}
        \centering
        \includegraphics[width=\linewidth]{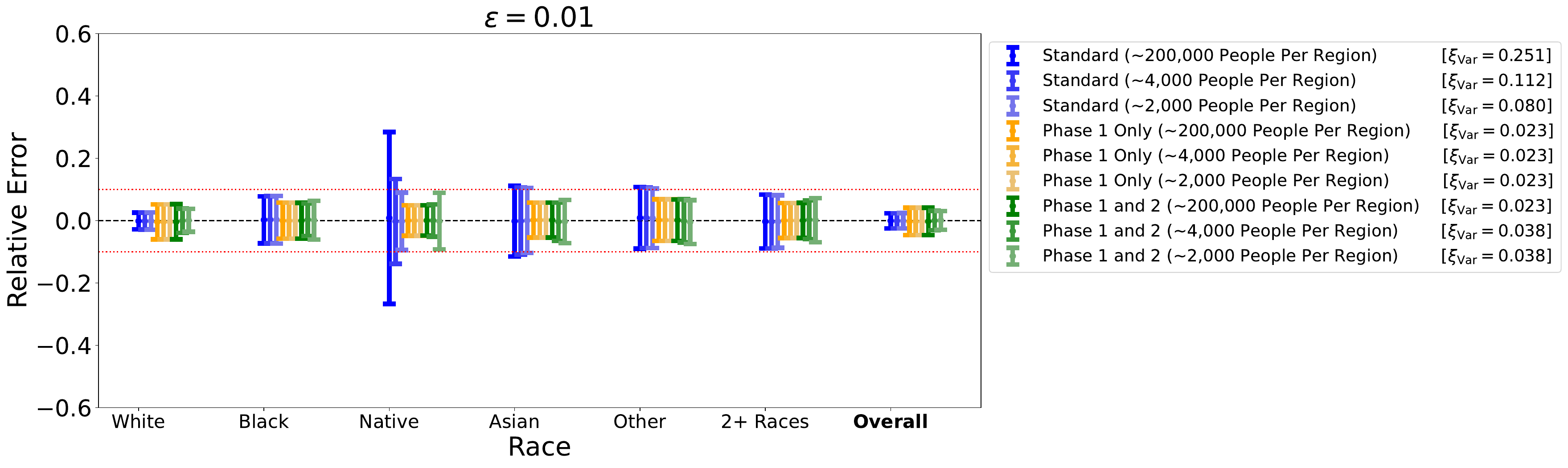}
    \end{subfigure}
    \caption{Relative errors from estimating mean income using $\varepsilon \in \{\infty, 1, 0.1, 0.01\}$ DP-noised $N_i^r$ with various region sizes (for Phase 2) in Connecticut. Regions with 200,000 PPR (People Per Region) are comparable to dividing the state into counties, 4,000 PPR is similar to dividing the state into Census Tracts, and 2,000 PPR resembles dividing the state into Census Block Groups.}
    \label{fig:sparsity}
\end{figure*}

\end{document}